\documentclass{vldb}

\usepackage{booktabs} 

\usepackage{latexsym}
\usepackage{newlfont}
\usepackage{algorithm}
\usepackage{algpseudocode}
\usepackage{epic}
\usepackage{eepic}
\usepackage{epsfig}
\usepackage{float}
\usepackage{graphicx}
\usepackage{caption}
\usepackage[caption=false]{subfig}
\usepackage{setspace}
\usepackage{breakurl}
\usepackage{url}
\usepackage{balance}
\usepackage{placeins} 

\usepackage{caption}
\algnewcommand{\LineComment}[1]{\State \(\triangleright\) #1}
\algdef{SE}{PForeach}{EndPForeach}[1]{\textbf{pforeach} \(\mbox{#1}\) \textbf{do}}{\textbf{end pforeach}}%

\newcommand{\COMMENT}[1]{}

\newtheorem{thm}{Theorem}
\newtheorem{obs}{Observation}
\newtheorem{lemma}{Lemma}

\newcommand{\RFULLRED}{rFullRed}
\newcommand{\RSMARTRED}{rSmartRed}
\newcommand{\NORED}{NoRed}

\newcommand{\PTOP}{pTop}
\newcommand{\PSMARTRED}{pSmartRed}

\newcommand{\skewedSet}{Whole}
\newcommand{\moreSkewedSet}{Skewed}
\newcommand{\mostSkewedSet}{MostSkewed}

\begin{document}

\title{Tail-Tolerant Distributed Search}
\numberofauthors{3}
\author{
\alignauthor
Naama Kraus\\
	\affaddr{Viterbi EE Department}\\
    \affaddr{Technion, Haifa, Israel}\\
\alignauthor
David Carmel\\
	\affaddr{Yahoo Research}
\alignauthor
Idit Keidar\\
	\affaddr{Viterbi EE Department}\\
    \affaddr{Technion, Haifa, Israel}\\
    \affaddr{and Yahoo Research}
}


\maketitle

\begin{abstract}
Today's search engines process billions of online user queries a day over huge collections of data.
In order to scale, they
distribute query processing among many nodes,
where each node holds and searches over a subset of the index called \emph{shard}.
Responses from some nodes occasionally fail to arrive within a reasonable time-interval 
due to various reasons,
such as high server load and network congestion.
Search engines typically need to respond in a timely manner, 
and therefore skip such \emph{tail latency} responses, 
which causes degradation in search quality.
In this paper, we tackle response misses due to high tail latencies 
with the goal of maximizing search quality.

Search providers today use redundancy in the form of Replication 
for mitigating response misses, 
by constructing multiple copies of each shard and searching all replicas.
This approach is not ideal,
as it wastes resources on searching duplicate data.
We propose two strategies to reduce this waste.
First, we propose \emph{\RSMARTRED{}}, 
an optimal shard selection scheme for replicated indexes.
Second, when feasible, we propose to replace Replication with \emph{Repartition},
which constructs independent index instances instead of exact copies.
We analytically prove that \RSMARTRED{}'s selection is optimal for Replication,
and that Repartition achieves better search quality compared to Replication.
We confirm our results with an empirical study using two real-world datasets,
showing that \RSMARTRED{} improves recall compared to currently used approaches.
We additionally show that Repartition improves over Replication in typical scenarios.
\end{abstract}

{




\section{Introduction}

Commercial search engines serve tens of thousands of online queries a second, 
covering corpora with billions of documents. 
In order to scale with the massive data, a search service is typically deployed on a cluster of nodes, which jointly implement \emph{distributed search (DiS)}~\cite{Callan00distributedinformation,Google2003,WebSearchScalability2011,Earlybird2012,Maguro2013,JalapartiSIGCOMM13,YunSIGIR15,ZetaSocc12,Dean13}. 
A common DiS approach
is to partition the documents into subsets called \emph{shards}, 
where each shard is assigned to a node in which it is locally indexed and searched~\cite{Google2003,WebSearchScalability2011,Earlybird2012,Dean13}. 
In order to reduce the computational cost,
it is common to employ \emph{approximate search}~\cite{WebSearchScalability2011,Puppin2010,Graceful2011,Callan00distributedinformation,JalapartiSIGCOMM13,YunSIGIR15,ZetaSocc12,Dean13}, 
whereby queries are sent to only a subset of the shards 
that are deemed most likely to satisfy the query.

Search providers nowadays aim to deliver results to the client within a few hundreds of milliseconds, 
expecting processing times of tens of milliseconds from the back-end search service~\cite{timeout1,timeout2,ZetaSocc12,Dean13}. 
Unfortunately, responses from nodes may sometimes take excessively long to arrive; 
this occurs due to many reasons, including network or server load, 
background processing, misconfiguration, crashes, etc.
This phenomenon is called the \emph{tail latency} problem -- the problem that tail (high percentile) latencies are much larger than the average latency~\cite{JalapartiSIGCOMM13,YunSIGIR15,ZetaSocc12,Dean13}.
Given that an increase in search response time directly implies loss in revenue~\cite{JalapartiSIGCOMM13,Brutlag09,Velocity09},
search engines typically set strict timeouts and ignore late responses~\cite{ZetaSocc12,Dean13}.
Consequently, some of the relevant results are missed,
entailing degradation of search quality~\cite{JalapartiSIGCOMM13,ZetaSocc12,Dean13,Graceful2011}.

Because in large data centers high tail latencies 
are the norm~\cite{JalapartiSIGCOMM13,Dean13,YunSIGIR15,ZetaSocc12,Hotnets12}, 
commercial search engines often deploy shards over a replicated storage layer, 
which, to guarantee timely responses, directs queries to all replicas of the requested shard~\cite{Google2003,ElasticSearchRepLevel,SolrRepLevel,JalapartiSIGCOMM13,Hotnets12,Dean13,GomezPantojaMCB11,Graceful2011,WebSearchScalability2011}.  
Although Replication is the standard approach to building fault-tolerant services~\cite{DistComputingBook}, 
we observe that in the context of search it is not ideal. 
This is because the challenges in the two cases are different: whereas classical fault-tolerance is concerned with services that are either available or unavailable,  in search, the \emph{quality} of the results is of essence. 
In this context, accessing multiple replicas of the same shard can be wasteful~\cite{Hotnets12,JalapartiSIGCOMM13,Dean13}, 
most notably when result misses are infrequent; 
a better use of resources could be accessing additional shards instead of additional copies of the same shard.

In this paper we consider tail-tolerance~\cite{Dean13} as a first class citizen in distributed search design. 
We study the problem of \emph{tail-tolerant distributed search}, 
whose goal is to maximize search quality when responses are missed due to high tail latencies.  
We suggest two improvements over the standard approach. 
First, we present \emph{\RSMARTRED{}}, 
an optimal selection scheme for replicated distributed search indexes.
Given a query, \RSMARTRED{} considers each shard's probability to satisfy the query,
as well as the probability to miss the shard's results,
in determining the number of replicas to select per shard.
Second, when feasible, we propose to employ \emph{Repartition}, 
an alternative to Replication that reduces the waste due to searching redundant shards.
Repartition randomly constructs independent partitions of the index
instead of exact copies.

Following the seminal work of Lv et al.~\cite{Multi-probeLSH}, 
we use \emph{success probability} and \emph{recall} metrics for measuring search quality.
We analyze the success probability to find a document relevant to a given query 
using a different DiS algorithms;
we prove that \RSMARTRED{}'s selection scheme is optimal for Replication,
and that Repartition improves over Replication for any selection scheme.
We confirm our analysis by conducting an empirical study using 
the Reuters RCV1 and Livejournal real-world datasets.
Our experiments show that \RSMARTRED{} achieves higher recall than techniques used today.
We further show the superiority of Repartition over Replication when excessive latencies are infrequent.

The rest of this paper is organized as follows:
We start by reviewing prior art in Section~\ref{sec:rel}.
In Section~\ref{sec:model}, 
we present a model and problem definition for tail-tolerant DiS.
In Section~\ref{sec:redundancy}, 
we present \RSMARTRED{} and Repartition,
our two novel tail-tolerant distributed search strategies.
In Section~\ref{sec:analysis} we prove that \RSMARTRED{} is optimal for Replication and that Repartition outperforms Replication.
In Section~\ref{sec:emp} we detail our empirical study.
We conclude the paper in Section~\ref{sec:conc}.

\section{Related Work}
\label{sec:rel}
Fault-tolerant distributed systems have been extensively explored in the literature
and standard textbooks, e.g.,~\cite{DistComputingBook}.
This line of work typically considers a binary availability model,
where the service is either available or unavailable,
whereas distributed search's availability can be captured by a finer-grain notion of search quality.

A wealth of prior art explores distributed search
(refer to~\cite{Callan00distributedinformation,WebSearchScalability2011} for a comprehensive overview),
where the index is partitioned into shards and distributed among multiple nodes.
\COMMENT {
and a distributed query processing is applied at runtime.

Index partitioning can be either \emph{term-based},
whereby index subsets consist of postings lists of disjoint subsets of terms,
or \emph{document-based}, in which case index subsets consist of postings lists of disjoint 
subsets of documents. 
In the document-based approach, documents are partitioned either randomly or 
according to their characteristics such as quality, popularity, language, topic, 
and geographic location.
}
At runtime, a \emph{broker} handling the query distributes it to nodes,
awaits their responses, aggregates the responses, and sends the query result to the user. 
The broker may employ search over all shards,
but to avoid excessive computation, it is more common to use
approximate search~\cite{Callan00distributedinformation,WebSearchScalability2011,Kulkarni2015,Gionis99,JalapartiSIGCOMM13},
which selects a subset of the shards to search over.
The selection of shards to search over is either random~\cite{WebSearchScalability2011,JalapartiSIGCOMM13}, 
or based on the estimated likelihood of the shards to contain results that are relevant to the query~\cite{Callan00distributedinformation,Gionis99}.

Most existing academic work on distributed search does not consider tail-tolerance.
Nevertheless, high tail latency is a serious problem in practice, 
and so industrial solutions must take it into consideration~\cite{Dean13,ElasticSearchRepLevel,SolrRepLevel,JalapartiSIGCOMM13,YunSIGIR15,ZetaSocc12}.
The ubiquitous approach to dealing with high tail latencies is truncating the tail, 
namely, responding to the user without waiting for responses from all nodes~\cite{Dean13,JalapartiSIGCOMM13,YunSIGIR15,ZetaSocc12,Hotnets12}. 
The rationale behind this approach is that
``returning good results quickly is better than returning the best results slowly''~\cite{Dean13}.
To compensate for the omitted responses, it is common to use  
redundancy in the form of Replication~\cite{Google2003,ElasticSearchRepLevel,SolrRepLevel,JalapartiSIGCOMM13,Hotnets12,Dean13,GomezPantojaMCB11,Graceful2011,WebSearchScalability2011}.
Note that commercial search engines~\cite{Dean13,JalapartiSIGCOMM13,YunSIGIR15,ZetaSocc12} apply 
engineering decisions (typically architecture-specific)
and other optimizations to reduce the tail latency of the search workflow;
such strategies are orthogonal to our research, and are
not sufficient by themselves~\cite{Dean13,Hotnets12}. 
Replication complements these optimizations, and is the focus of this work.

Prior art observed that Replication incurs resource waste 
due to duplicate search operations~\cite{Hotnets12,JalapartiSIGCOMM13,Dean13}.
Commercial search engines~\cite{Dean13,JalapartiSIGCOMM13} decrease this waste
by combining two techniques.
First, they only re-issue a search request for slow shards,
and second, 
they cancel ongoing duplicate requests upon learning 
that they are not likely to contribute much to the search quality.
Although these strategies were shown to be useful to some extent,
the approach of simultaneously sending multiple copies of each request is still commonly used
despite the waste it incurs~\cite{Hotnets12,Dean13}.
In this paper, we tackle Replication's inherent waste using two improvements:
First, we propose a simple optimal algorithm for adjusting the replication level 
for each shard when processing a given query.
Second, we propose an alternative approach to redundancy, which further decreases waste.

\section{Model and Problem Definition}
\label{sec:model}
In this section, we provide a brief background on the primary components of a search system, 
and present the \emph{distributed search (DiS)} model that we consider.
We then present the problem of search quality degradation due to high tail latency in DiS,
and the search quality metrics that we use.

\subsection{Search}
\label{sec:search_basics}
The goal of a search system is to retrieve relevant documents
to a given query from a given document collection $\mathcal{D}$~\cite{Manning2008}.
Typically, a search service consists of two basic primitives:
\begin{description}
	\item[Indexing] pre-processes $\mathcal{D}$
    and indexes the documents into a persistent data structure called \emph{inverted index},
    which we refer to shortly as \emph{index} in this paper.
	\item[Query processing] 
    uses the index to retrieve a list of documents ranked according to their estimated relevance to the query.
\end{description}

Centralized  search, deployed on a single machine, does not scale with the size of $\mathcal{D}$~\cite{Google2003,WebSearchScalability2011},
and so is not used in practice for searching large data collections.

\subsection{Distributed search}
\label{sec:dist_search}
Distributed search scales the search service
by distributing indexing and query processing over multiple nodes~\cite{Callan00distributedinformation,Google2003,WebSearchScalability2011,Earlybird2012,Maguro2013,JalapartiSIGCOMM13,YunSIGIR15,ZetaSocc12,Dean13}.
We consider the common approximate search approach for distributed query processing~\cite{WebSearchScalability2011,Puppin2010,Graceful2011,Callan00distributedinformation,JalapartiSIGCOMM13,YunSIGIR15,ZetaSocc12,Dean13},
which in order to reduce computational costs, 
submits each query to only
a selected subset of the nodes. 
Note that in order to avoid missing important (e.g., popular) results,
search engines commonly dedicate certain nodes to storing important documents,
which are always searched~\cite{WebSearchScalability2011,Earlybird2012};
we consider here an approximate search over the rest of the nodes.

We assume a cluster of $n>1$ nodes connected by a fast network.
We consider the common \emph{document-based} approach to DiS, 
whereby each node holds and performs search on some subset of $\mathcal{D}$~\cite{Google2003,WebSearchScalability2011,Earlybird2012,Dean13}.
To implement DiS, one needs to address two aspects:
partitioning and shard selection.
Figure~\ref{fig:dis_architecture} illustrates DiS architecture, which we further detail in the next paragraphs.

\begin{figure}[hbt]
	\centering
    \includegraphics[scale=0.25]{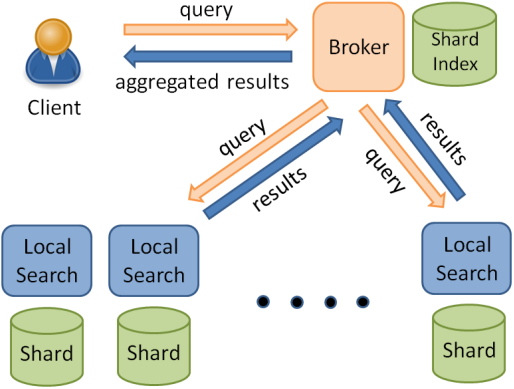}
    \vspace{1em}
     \caption{Illustration of distributed search (DiS) architecture supporting approximate search. 
              At indexing time, DiS partitions the document collection into shards,
              which are distributed among the nodes.
              At runtime, a broker module serves clients' queries:
              for each query, the broker first selects a subset of the shards to search over.
              Each node searches its shard locally and returns its search results to the broker.
              Finally, the broker aggregates the nodes' results and returns them to the client.}
	\label{fig:dis_architecture}
\end{figure}

\paragraph*{Partitioning}
At the indexing stage,
DiS applies a \emph{partitioning scheme} to partition $\mathcal{D}$
into a set of $n$ pairwise disjoint \emph{shard}s $D_j \subset \mathcal{D}$, 
$\mathcal{D} = \bigcup_{j=1}^n D_j$.
Each shard $D_j$ is assigned to a separate node where it is locally indexed.
One common approach is similarity-based partitioning~\cite{WebSearchScalability2011,Callan00distributedinformation,Puppin2010,Kulkarni2015},
e.g., LSH~\cite{LSH_Indyk1998,Gionis99}, which constructs shards of similar documents.
LSH randomly selects a hash function that
maps each document $d \in \mathcal{D}$ into its corresponding shard,
where the probability that two documents are mapped to the same shard
grows with their similarity~\cite{Charikar02,Chierichetti12}.
Note that LSH is a natural choice for distributed partitioning,
since parallelizing it
is straightforward because hashing a document does not require any information about other documents.

\paragraph*{Shard selection}
A broker module handles clients' search requests~\cite{WebSearchScalability2011}.
At runtime, the broker 
accepts an input query and submits the query to the nodes,
each of which locally searches its shard  
and returns to the broker a ranked list of results that it finds most relevant to the query. 
The broker then collects and merges the shards' results
and returns a final result set to the caller.

The broker uses a \emph{shard selection scheme} in order
to select a subset of $t \leq n$ 
shards to send the query to. 
Typically, the
selection scheme uses a \emph{shard index},
which holds the partition's meta-data.
More specifically, it maps shard identifiers to the nodes where they reside,
and optionally maintains some compact representation of each shard's content.
The shard index is constructed during the indexing stage,
it is typically centralized and replicated for availability.

Given a query, the shard selection scheme approximates 
a probability distribution 
over the shards,
associating with each shard $D_j$ the estimated probability that it contains a relevant document to the query;
the latter is called the shard's \emph{success probability} for the query.

Shard selection may use the simple \emph{Random}~\cite{WebSearchScalability2011,JalapartiSIGCOMM13} approach, 
which does not employ a shard representation, and
randomly selects shards independently of the input query.
Note that Random induces a uniform success probability distribution over the shards.
A more effective approach is 
to select shards that are deemed most likely to satisfy the query
(refer to~\cite{WebSearchScalability2011} for an overview of selection methods).
One popular method is ReDDE~\cite{REDDE},
which represents a shard using a random sample of its content.
At the indexing stage, 
ReDDE randomly samples documents from each shard
and indexes them into the shard index\footnote{In ReDDE, the shard index is commonly called a centralized sample index (CSI).}.
At shard selection time, given a query $q$, 
ReDDE retrieves from the shard index a set of 
documents that are most relevant to $q$,
and based on them, selects $t$ shards that are 
most likely to contain documents relevant to $q$.
In our experiments, we use CRCS Linear~\cite{Shokouhi2007}
to approximate the success probability distribution over the shards,
which we detail in Section~\ref{sec:emp}.

\COMMENT {
\paragraph*{CORI}
CORI represents a shard using a centroid vector of the shard's documents,
and estimates the probability that it contains relevant documents
according to the similarity of its centroid to the query.
To create the representation vector,
CORI weights each term according to its \emph{term frequency (TF)} and \emph{inverse collection frequency (ICF)},
where $TF$ is proportional to the term's frequency in the shard,
and $ICF$ is proportional to the number of shards that contain the term.
CORI then indexes the centroid representation for each shard into the centralized shard index.
Given a query, CORI searches the shard index and selects the shards 
that are most similar to the query.
}

\subsection{The impact of high tail latency}
\label{sec:ftSearch}
We now extend DiS to consider high tail latency.
In order to provide search results in a timely manner (typically a search latency of few hundreds of milliseconds ~\cite{timeout1,timeout2,Dean13,ZetaSocc12}),
the broker waits for responses from nodes up to a fixed timeout
that is given to DiS as a parameter~\cite{ZetaSocc12,Dean13}.
The broker collects results from the nodes that respond on time,
and drops the results of the slow nodes~\cite{JalapartiSIGCOMM13,ZetaSocc12,Dean13}.
A node may fail to return its results with the desired latency due to various reasons: 
E.g., it may be temporarily down due to hardware or software problems, be overloaded by other queries, 
or lose messages due to network failures or loads~\cite{Dean13}.
We assume that each node fails to respond on time with some \emph{miss probability}.
For simplicity, we assume that each node fails to respond independently of other shards,
and that the miss probability is common to all nodes.
We denote the miss probability by $f$.
When a node's response is skipped, some of the relevant results may be missing from the final result set, 
which entails degradation of search quality~\cite{JalapartiSIGCOMM13,ZetaSocc12, Dean13,Graceful2011}.
Note that the search quality of DiS with approximate search is typically lower than that of centralized search 
even without misses, as the search is restricted to a subset of the collection. 
Result misses due to high tail latencies further degrade search quality;
our goal is to ameliorate this.

\subsection{Search quality metrics}
\label{sec:metric}
We analyze the quality of a DiS algorithm $A$ through its
\emph{success probability}:
Given a query $q$ and a unique document $d_q \in \mathcal{D}$ relevant to $q$,
$SP(A,q)$ is the probability that $A$ finds $d_q$.

We empirically measure the search quality of a DiS algorithm $A$ 
by comparing its results with those of centralized 
search, which has full access to all documents~\cite{Puppin2010,Multi-probeLSH,Graceful2011}.
More specifically, let $S^m_C(q)$ be the top-m search results for query $q$ 
according to the centralized search, 
and let $S^m_A(q)$ be the top-m results of a DiS algorithm $A$. 
We measure the search quality of $A$ for query $q$ by the \emph{recall} it achieves relative to
a centralized system:
\[
Recall@m(q) \triangleq \frac{|S^m_C(q)\cap S^m_A(q)|}{|S^m_C(q)|}.
\]
Note that $Recall@m(q) \in [0,1]$ since $S^m_A(q) \subseteq S^m_C(q)$.
We measure the search quality of a DiS algorithm $A$, 
$Recall@m$, by averaging $Recall@m(q)$ over all queries.

\section{Tail-Tolerant DiS}
\label{sec:redundancy}
Existing DiS systems mitigate search quality degradation using
Replication~\cite{Google2003,ElasticSearchRepLevel,SolrRepLevel,JalapartiSIGCOMM13,Hotnets12,Dean13,GomezPantojaMCB11,Graceful2011,WebSearchScalability2011}.
We propose two improvements to currently used approaches:
\RSMARTRED{},
an optimal shard selection scheme for Replication,
and Repartition,
an alternative method to redundancy which improves over Replication.

\subsection{Replication}
\label{sec:model_replication}
Given a redundancy level configuration parameter $r>1$ and a partition $\left\{D_1, \ldots, D_n\right\}$,
Replication constructs $r$ identical copies of that partition.
Large-scale search systems usually deploy their service over multiple data centers and use dozens of replicas~\cite{Google2003},
whereas smaller-scale search systems that run in a single data center 
typically use a few copies per shard, e.g., $r=3$~\cite{ElasticSearchRepLevel,SolrRepLevel}.
A shard selection scheme that uses Replication needs to take an additional aspect into account:
Besides identifying the shards most likely to satisfy the query, 
it needs to also decide how many replicas of each shard to contact.
We discuss three approaches for doing so.
We assume that all approaches are given a fixed budget of $tr$ shards to select 
out of all $nr$ shards.

\subsubsection{Existing shard selection approaches}
Two main approaches are used for shard selection today.
First, in some cases,
redundancy is used only for load-balancing and not for mitigating result misses~\cite{WebSearchScalability2011,Earlybird2012},
yielding an approach we call ``no redundancy'', denoted \emph{\NORED{}}.
\NORED{} selects all $tr$ shards from a single partition without replicas
($tr \le n$),
and the broker directs different queries to different index partitions.
Figure~\ref{fig:ftdis}\subref{subfig:SinglePartition} illustrates \NORED{}.

In other cases, a ``full redundancy'', denoted \emph{\RFULLRED{}},
approach is used~\cite{GomezPantojaMCB11,Graceful2011,WebSearchScalability2011,Google2003,Junqueira2012,ElasticSearchRepLevel,SolrRepLevel,Dean13,JalapartiSIGCOMM13,Hotnets12}. 
Given a query, the broker selects $t$ out of $n$ shards of the original partition,
and replicates its selection by contacting all $r$ replicas of each selected shard.
Figure~\ref{fig:ftdis}\subref{subfig:RepBase} illustrates \RFULLRED{}.
This approach arises when shard selection and redundancy are two separate abstraction layers,
and the search algorithm uses replicated storage as a black box.

\begin{figure}[hbt]
	\centering
    \subfloat[\NORED{}]{\label{subfig:SinglePartition}\includegraphics[scale=0.25]{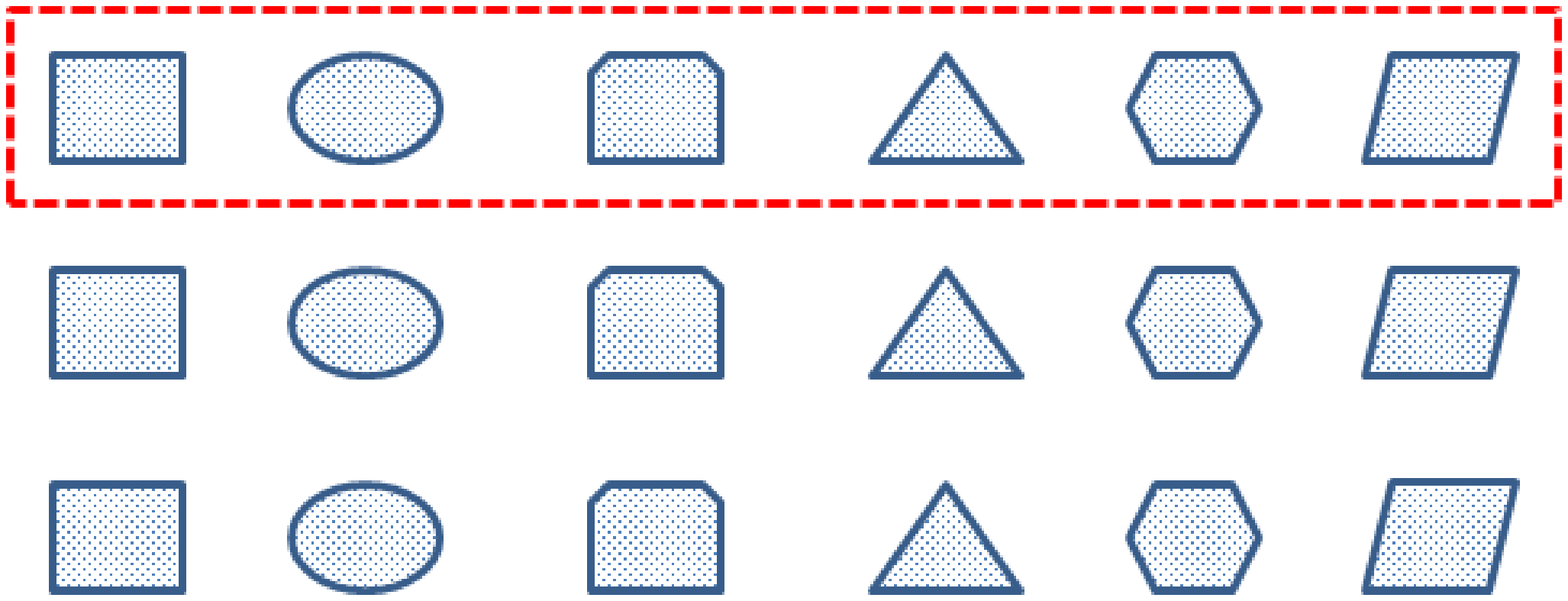}}\\\vspace{1em}
    \subfloat[\RFULLRED{}]{\label{subfig:RepBase}\includegraphics[scale=0.25]{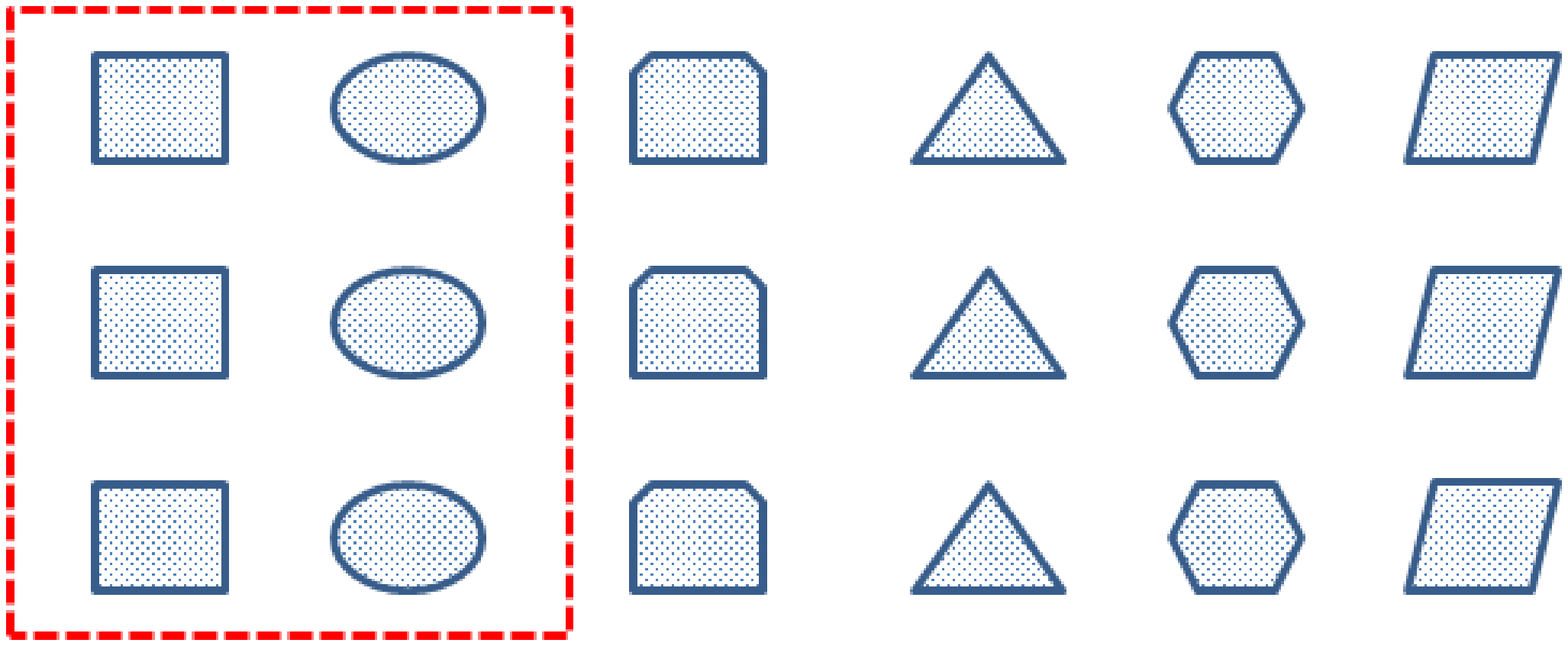}}\\\vspace{1em}
	\subfloat[\RSMARTRED{}]{\label{subfig:RepOpt}\includegraphics[scale=0.25]{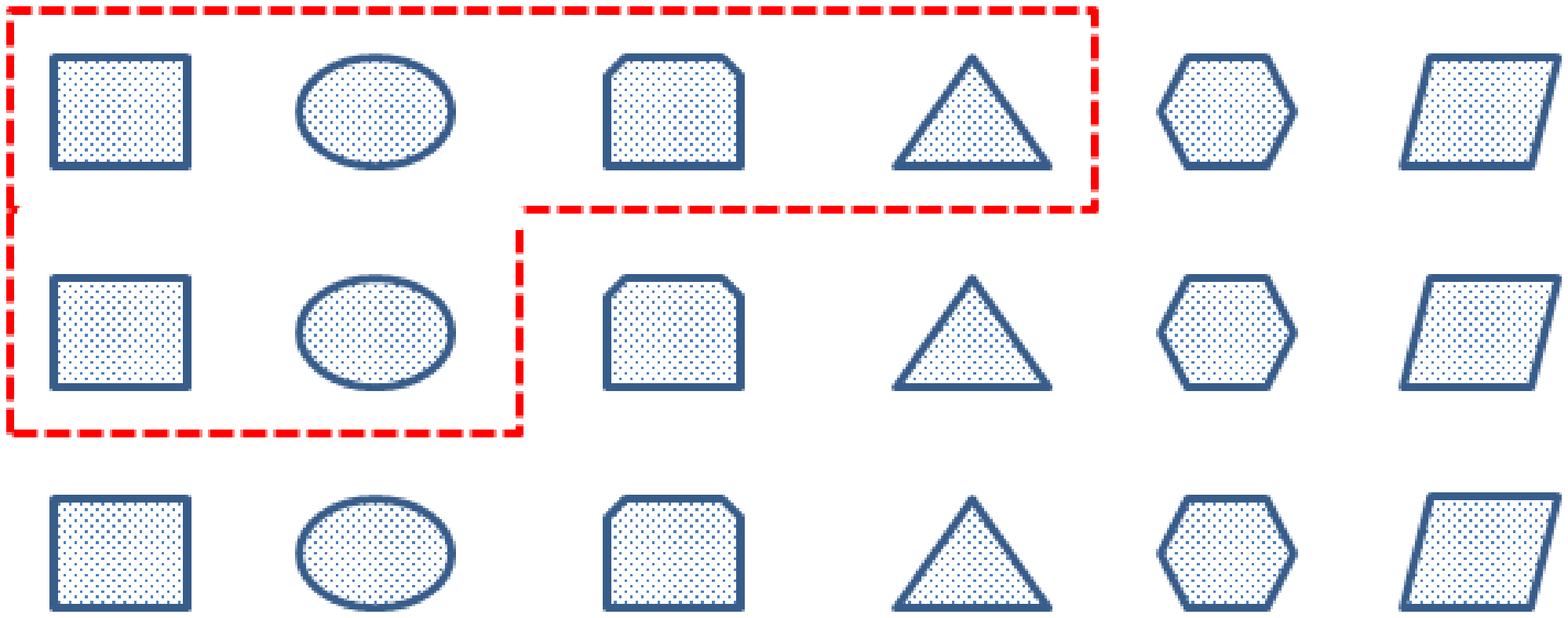}}   
    \vspace{2em}
     \caption{Illustration of the three shard selection methods under Replication.
              A row represents a partition of the collection into $6$ shards.
              At runtime, a total of $6$ shards (dashed lines) are selected.
              \NORED{} selects shards from a single partition without replicas; 
              \RFULLRED{} contacts all $3$ replicas of each shard it selects; 
              \RSMARTRED{} adjusts the number of selected shard replicas such that search quality is maximized.}
	\label{fig:ftdis}
\end{figure}

\subsubsection{Optimal shard selection}
\label{sec:optRep}
We next open up the black box and integrate shard selection and redundancy.
We present an optimal approach for replicated shard selection, \RSMARTRED{},
which maximizes the probability to find relevant documents.
Our method considers both the miss probability and the success probability distribution when selecting shard replicas.

To give an intuition why existing approaches are not optimal, 
consider the following example:
the dataset is partitioned into $5$ shards, 
each shard has two replicas ($r=2$),
and the broker selects $tr=2$ shards per query.
For some query $q$, $D_1$'s success probability is $0.8$,
and $D_2$'s success probability is $0.1$. 
The success probability of the rest of the shards is smaller.
Clearly, $D_1$ should be selected at least once.
There are two alternatives for selecting the second shard:
$D_1$'s replica or $D_2$. 
If $D_1$'s two replicas are selected,
a relevant document $d_q$ is found if it is stored in $D_1$,
and at least one of $D_1$'s replicas does not fail to respond,
which happens with probability $0.8(1-f^2)$.
If $D_1$ and $D_2$ are selected,
$d_q$ is found if it is stored in either $D_1$ or $D_2$,
and the shard that contains it does not fail to respond.
As $D_1$ and $D_2$ are disjoint, this happens with probability $(0.8+0.1)(1-f)$.
Table~\ref{table:selectionIntuition} depicts the success probabilities of the two selection alternatives
for two values of $f$.
As the table demonstrates, the selection that maximizes the success probability depends on the value of $f$.
For $f=0.05$, selecting $D_1$ and $D_2$ is preferable,
whereas for $f=0.2$, selecting the two replicas of $D_1$ is preferable.

\begin{table}[hbt]
\begin{center}
\begin{tabular}{|c|c|c|}
\hline
       		  & Two replicas of $D_1$    & $D_1$ and $D_2$ \\ \hline 
$f=0.05$      & $0.8$                    & $0.85$           \\ 
$f=0.2$       & $0.77$                   & $0.72$           \\ \hline
\end{tabular}\vspace{2em}
\caption{Success probability of different shard selections (columns) for different miss probabilities (rows)
	    when selecting a total of $tr=2$ shards under Replication.
        When $f=0.05$, it is preferable to select the top two shards from the same partition,
        whereas when $f=0.2$, it is preferable to select two replicas of the highest ranked shard.
        }
\label{table:selectionIntuition}
\end{center}
\end{table}

\COMMENT{
\begin{figure}[hbt]
    \centering
    \includegraphics[scale=0.5, clip]{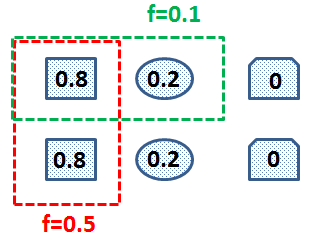}
    \caption{Alternatives for selecting shards under Replication,
              when selecting a total of $tr=2$ shards.
              As shown in Table~\ref{table:selectionIntuition},
              when $f=0.1$, it is preferable to select the top two shards from the same partition,
              whereas when $f=0.5$, it is preferable to select two replicas of the highest ranked shard.}
						\label{fig:fwpAndFailures}
\end{figure}
}

Our \RSMARTRED{} algorithm considers 
$f$ and an estimated distribution of the shard success probabilities.
Given a query $q$, we denote by $p_q(j)$ the estimated success probability of shard $D_j$.
Given $r$ replicas of the partition,
we assign a score of $f^{i-1} p_q(j)$ to the $i$th replica of shard $D_j$
as depicted in Table~\ref{table:SmartRedScores}. 
We then select the $tr$ shard replicas with the highest scores.
Figure~\ref{fig:ftdis}\subref{subfig:RepOpt} shows an example selection of
\RSMARTRED{}.
In Section~\ref{sec:analysis},
we prove that if the estimations are accurate, then 
\RSMARTRED{}
maximizes the probability to find relevant documents 
with any given number of selections.
Note that \RSMARTRED{} is more likely to select multiple 
replicas of a shard as the shard's success probability increases and as $f$ increases.

\begin{table}[hbt]
\begin{center}
\begin{tabular}{|c|ccc|}
\hline
       		   & $D_1$     & $\ldots$     & $D_n$  \\ \hline 
Replica $1$   & $p_q(1)$  & $\ldots$     & $p_q(n)$  \\ 
Replica $2$   & $p_q(1)f$  & $\ldots$     & $p_q(n)f$  \\ 
$\vdots$       & $\vdots$  &               & $\vdots$ \\
Replica $r$   & $p_q(1)f^{r-1}$ &          & $p_q(n)f^{r-1}$  \\ \hline
\end{tabular}
\vspace{2em}
\caption{Scores of shard replicas in \RSMARTRED{}.
         \RSMARTRED{} selects the $tr$ shard replicas with the highest scores.}
\label{table:SmartRedScores}
\end{center}
\end{table}

The following observation follows immediately from the algorithm:
\begin{obs}
\label{obs:topPerPartition}
For each $i$, 
the $t_i=|S_i|$ shards that \RSMARTRED{} selects from partition $i$ are the $t_i$ top-scored shards in that partition according to success probability.
\end{obs}

\subsection{Repartition}
\label{sec:perturbation}
We propose Repartition, a new approach for constructing a redundant index for tail-tolerant DiS.
Like Replication, Repartition constructs $r$ partitions of $\mathcal{D}$,
each consisting of $n$ pairwise disjoint shards,
and each of the $nr$ shards is assigned to a separate node.
However, unlike Replication, 
Repartition does not construct exact copies of the partition.
Instead, Repartition uses a randomized partitioning scheme and applies it
$r$ times independently to construct $r$ independent partitions.
One may use LSH~\cite{LSH_Indyk1998,Gionis99} for implementing Repartition
by randomly and independently selecting $r$ hash functions.
We propose two shard selection schemes for Repartition: \emph{\PTOP{}} and \emph{\PSMARTRED{}}. 
Note that \NORED{} is trivially applicable to Repartition as well.

\paragraph*{\PTOP{}}
As partitions are independent,
a natural shard selection scheme for Replication selects the $t$ top-scored shards from each 
partition independently, 
where a shard's score is its estimated success probability.
We call this selection scheme \PTOP{}.
Note that like \RFULLRED{}, \PTOP{} selects the same number of shards ($t$) from each partition.

\paragraph*{\PSMARTRED{}}
Our second selection, \PSMARTRED{}, imitates \RSMARTRED{}, 
and  works as follows:
\PSMARTRED{} first arbitrarily selects one of the partitions of $\mathcal{D}$, 
and computes \RSMARTRED{}'s shard selection over $r$ replicas of $\mathcal{D}$.
Recall that \RSMARTRED{} selects $t_i$ shards from each partition replica $i$.
\PSMARTRED{} then selects the $t_i$ top-scored shards from each partition $i$ of the re-partitioned index 
according to the success probability distribution of the shards in that partition.
Therefore, \PSMARTRED{} preserves the number of shards that \RSMARTRED{} selects from each partition.
For example, 
\PSMARTRED{} applies \RSMARTRED{}'s selection that is illustrated in Figure~\ref{fig:ftdis}\subref{subfig:RepOpt},
by selecting the four top-scored shards from one partition 
and the two top-scored shards from the second.

As we show both analytically and in our empirical study, Repartition improves over Replication.
On the other hand, creating and maintaining the index are
more costly with Repartition.
Another limitation of Repartition
is that it is not applicable when the partitioning is given by a third party and cannot be altered.

\COMMENT {
LSH is a randomized algorithm widely used for 
partitioning documents by similarity.
Given document collection $\mathcal{D}$ and a similarity function 
between documents 
$sim: \mathcal{D} \times \mathcal{D} \rightarrow [0,1]$~\cite{Chierichetti12},
LSH is a distribution over a family of hash functions $\mathcal{G}: \mathcal{D} \rightarrow \left\{D_1, \ldots, D_n\right\}$
that maps documents into shards,
where the probability that a randomly sampled 
$g \in \mathcal{G}$ maps two documents to the same shard
grows with their similarity~\cite{Charikar02,Chierichetti12}.
More formally,
given two documents $d_1,d_2 \in \mathcal{D}$,
under a random selection of $g \in \mathcal{G}$,
\begin{equation}
Pr_{g \in \mathcal{G}}[g(d_1)=g(d_2)]= (sim(d_1,d_2))^k,
\end{equation}
where $k$ is a configuration parameter than controls $\mathcal{G}$'s precision.

In order to partition $\mathcal{D}$ into shards of similar documents,
LSH randomly selects a hash function $g \in \mathcal{G}$,
and maps each document $d \in \mathcal{D}$ into its corresponding shard $g(d)$.
Note that LSH is a natural choice for distributed partitioning,
since parallelizing it
is straightforward because hashing a document does not require any information about other documents.
We construct $r$ independent partitions of $\mathcal{D}$ using
$r$ randomly and independently selected hash function $g_i \in \mathcal{G}$.
}

\section{Analysis}
\label{sec:analysis}
\begin{sloppypar}
In this section, we analytically study the success probability~\cite{Multi-probeLSH} 
to retrieve a document relevant to a query when searching a tail-tolerant distributed index.
We provide closed-form analysis of the success probability under Replication,
and prove that \RSMARTRED{} is the optimal selection.
We also prove that
for any shard selection scheme for Replication, 
there exists a shard selection scheme for Repartition with a larger or equal success probability.

\subsection{Success probability formulation}
\label{sec:spFormulation}
Consider a tail-tolerant DiS algorithm $A(r,t)$ with redundancy $r$ and $tr$ shards selected per query.
Consider a query $q$.
Although multiple documents may be relevant to $q$,
for the sake of the analysis, 
we consider exactly one document $d_q \in \mathcal{D}$ that is relevant to $q$.
For query $q$,
the shard selection scheme induces a probability distribution  
$p_q: \left\{1,\ldots,n\right\} \rightarrow [0,1]$,
where $p_q(j)$ is the probability that $d_q$ is stored in shard $D_j$
(in practice, shard selection schemes such as ReDDE approximate this distribution).
Since $\left\{D_1, \ldots, D_n\right\}$ is a partition, 
$\sum_{j=1}^n p_q(j) = 1$.
We denote by $SP(q,f,A(r,t))$ the probability that $A(r,t)$ 
finds $d_q$ when processing query $q$ under miss probability $f$. 
$SP$ is called $A(r,t)$'s success probability.
Henceforth we fix a query $q$ and remove $q$ from our notations.

\subsection{Replication}
\label{sec:replication}
Consider a replicated DiS algorithm $A_R(r,t)$.
We denote by $S_i \subseteq \left\{D_1, \ldots, D_n\right\}$ the set of shards for which $A_R(r,t)$ selects at least $i \geq 1$ replicas.
For example,
if three replicas of $D_7$ are selected,
then $D_7 \in S_1, S_2, S_3$.
Note that $\sum_{i=1}^r \left|S_i\right| = tr$. 
In addition,
\begin{equation}
\label{eq:contain}
S_r \subseteq S_{r-1} \subseteq \ldots \subset S_1.
\end{equation}
We denote by $SP(f,S_i)$ the probability that $d_q$ is found when 
accessing the shards in $S_i$.
This occurs if $d_q$ is stored in one of the shards in $S_i$ and 
that shard
does not fail to respond.
Since shards in $S_i$ are disjoint:
\begin{equation}
\label{eq:probRow}
SP(f,S_i) = (1-f) \sum_{D_j \in S_i} p(j). 	
\end{equation}

By definition, $A_R$'s success probability equals $SP(f,\bigcup_{i=1}^r S_i)$.
The following lemma formulates $SP(f,\bigcup_{i=1}^r S_i)$:
\begin{lemma}
\label{lemma:repRrows}
\[
SP(f,\bigcup_{i=1}^r S_i) =  (1-f) \left(\sum_{D_j \in S_1} p(j) + \ldots \sum_{D_j \in S_r} f^{r-1}p(j)\right).
\]
\end{lemma}

\begin{proof}
We prove by induction on $r$.
\paragraph*{Base ($r=1$):}
Follows directly
from Equation~\ref{eq:probRow}.

\paragraph*{Step:}
we assume for $r-1$: 
\begin{equation}
\label{eq:inducAssum}
SP(f, \bigcup_{i=1}^{r-1} S_i) = (1-f) \left(\sum_{D_j \in S_1} p(j) + \ldots + \sum_{D_j \in S_{r-1}} f^{r-2}p(j)\right).
\end{equation}

As $\bigcup_{i=1}^r S_i = (\bigcup_{i=1}^{r-1} S_i) \cup S_r$,
then according to the probability of a union of events\footnote{$Pr(A \cup B) = Pr(A) + Pr(B) - Pr(B)Pr(A|B).$}:
\[
SP(f, \bigcup_{i=1}^r S_i)= SP(f,\bigcup_{i=1}^{r-1} S_i) + SP(f,S_r) 
\]
\begin{equation}
- SP(f, S_r) SP((f,\bigcup_{i=1}^{r-1} S_i) | (f, S_r)),
\label{eq:unionEvents}
\end{equation}
where $SP((f,\bigcup_{i=1}^{r-1} S_i) | (f, S_r))$ denotes the conditional probability to find $d_q$ when searching
$\bigcup_{i=1}^{r-1} S_i$, 
given that $d_q$ is found when searching $S_r$. 
Let $D_j \in S_r$ be
the shard replica that contains $d_q$.
Due to containment (Equation \eqref{eq:contain}),
$D_j \in S_i$, $1 \le i \le r-1$.
Hence, $d_q$ is found if at least one of those $r-1$ shards does not fail to respond,
which happens with probability $1-f^{r-1}$.
Thus, 
\begin{equation}
\label{eq:cond}
SP((f,\bigcup_{i=1}^{r-1} S_i) | (f, S_r)) = 1-f^{r-1}.
\end{equation}

By Equations~\ref{eq:unionEvents} and ~\ref{eq:cond}:
\[
SP(f,\bigcup_{i=1}^r S_i) = SP(f,\bigcup_{i=1}^{r-1} S_i) + SP(f, S_r) 
\]
\begin{equation}
- SP(f, S_r) (1-f^{r-1}).
\label{eq:unionEventsRep}
\end{equation}

And by Equations~\ref{eq:probRow},~\ref{eq:inducAssum}:
\[
SP(f,\bigcup_{i=1}^r S_i)= (1-f) (\sum_{D_j \in S_1} p(j) + \ldots + \sum_{D_j \in S_{r-1}} f^{r-2}p(j))
\]
\[
+ (1-f)\sum_{D_j \in S_r} p(j) - (1-f^{r-1})(1-f)\sum_{D_j \in S_r} p(j)
\]
\[
= (1-f)\left(\sum_{D_j \in S_1} p(j) + \ldots + \sum_{D_j \in S_r} f^{r-1}p(j)\right).
\]
\end{proof}

\subsubsection{Optimal selection}
\begin{thm}
\label{thm:rsmartred}
For a given $1 \leq tr \leq nr$,
\RSMARTRED{} selects $tr$ shards such that $SP$ is maximized.
\end{thm}

\begin{proof}
According to Lemma \ref{lemma:repRrows}, $SP$ is maximized when
$\sum_{D_j \in S_1} p(j) + \ldots + \sum_{D_j \in S_r} f^{r-1}p(j)$ is maximized.
Selecting the $tr$ shards with the largest 
score values, $f^{i-1}p(j)$, maximizes the sum hence the success probability.
\end{proof}

\RSMARTRED{}'s selection depends on the miss probability and the shard success probability distribution
as formulated in Lemma \ref{lemma:repRrows}.
When $f$ is high and the distribution is skewed, i.e., few shards have a high success probability,
the optimal selection is likely to select those shards' replicas and tends towards the \RFULLRED{} method.
When the success probability distribution is close to uniform, or when $f$ is low,
the optimal selection is more likely to select additional shards of the partition,
hence it tends towards the NoRed method.
Note that $f^{i-1}p_{j}$ exponentially decreases as $i$ increases,
hence the effectiveness of selecting additional replicas of a shard decreases 
as more of its replicas are selected.

\subsection{Repartition}
\label{sec:analysis_repartition}

\begin{thm}
\label{thm:repartition}
Consider a Repartition in which all partitions have the same probability distribution
for a given query $q$.
Then, for every shard selection used with Replication of one of these partitions, 
there exists a shard selection for Repartition with a larger or equal success probability for $q$.
\end{thm}

\begin{proof}
Consider a Repartition consisting of $r$ independent 
partitions of $\mathcal{D}$,
and a Replication algorithm  $A_R(r,t)$
employing $r$ replicas of one of these partitions.
Consider $A_R(r,t)$'s shard selection, $S_1,\ldots,S_r$,
which selects $t_i \triangleq \left|S_i\right|$ top-scored shards from each partition replica $i$.
We construct a Repartition algorithm $A_P(r,t)$ to select the $t_i$ 
top-scored shards for each independent partition $i$, 
according to success probability.
I.e., $A_P(r,t)$ preserves the number of shards that $A_R(r,t)$ selects per each partition.
We denote $A_P(r,t)$'s selection by $S'_1,\ldots,S'_r$.
We prove that 
\begin{equation}
\label{eq:repProws}
SP_P(f,\bigcup_{i=1}^r S'_i) \ge  SP_R(f,\bigcup_{i=1}^r S_i).
\end{equation}
by induction on $r$.

\paragraph*{Base ($r=1$):}
According to Equation~\ref{eq:probRow}, and since the distributions are equal for all partitions:
$
SP_P(f,S'_1) =  SP_R(f,S_1).
$

\paragraph*{Step:}
we assume for $r-1$: 
\begin{equation}
\label{eq:inducAssumP}
SP_P(f,\bigcup_{i=1}^{r-1} S'_i) \ge  SP_R(f,\bigcup_{i=1}^{r-1} S_i).
\end{equation}

We compute for $r$:
as $\bigcup_{i=1}^r S'_i = (\bigcup_{i=1}^{r-1} S'_i) \cup S'_r$,
and as the partitions are independent,
then according to the probability of a union of 
independent events:
\[
SP_P(f,\bigcup_{i=1}^r S'_i)= SP_P(f,\bigcup_{i=1}^{r-1} S'_i) + SP_P(f, S'_r) 
\]
\begin{equation}
\label{eq:spP}
- SP_P(f, S'_r) SP_P(f,\bigcup_{i=1}^{r-1} S'_i).
\end{equation}
$SP_P(f,\bigcup_{i=1}^{r-1} S'_i)$
denotes the probability that we found $d_q$ when searching
$\bigcup_{i=1}^{r-1} S'_i$.
Denote by $\eta$ the probability that $d_q$ is stored in $\bigcup_{i=1}^{r-1} S'_i$.
For each $S'_i$, $1 \le i \le r-1$, $d_q$ is either stored in one of the shards in $S'_i$ or it is not.
Hence, 
$d_q$ is stored in $u$ shards in $\bigcup_{i=1}^{r-1} S'_i$,
where $0 \le u \le r-1$,
and is found if at least one of those shards does not fail to respond
which happens with probability $1-f^u$. 
Thus, 
\begin{equation}
\label{eq:spPrminus1}
SP_P(f,\bigcup_{i=1}^{r-1} S'_i) = \eta (1-f^u).
\end{equation}
Since both $A_R(r,t)$ and $A_P(r,t)$ select $t_r$ top-scored shards from their $r$-th partition,
and since success probabilities are equal for all partitions,
then according to Equation~\ref{eq:probRow},
\begin{equation}
\label{eq:eqR}
SP_P(f,S'_r) =  SP_R(f,S_r).
\end{equation}
Substituting Equations~\ref{eq:spPrminus1} and~\ref{eq:eqR} in
Equation~\ref{eq:spP} we get:
\[
SP_P(f,\bigcup_{i=1}^r S'_i) = SP_P(f,\bigcup_{i=1}^{r-1} S'_i) + SP_R(f, S_r) 
\]
\[
- SP_R(f, S_r) \eta (1-f^u).
\]
According to the induction assumption in Equation~\ref{eq:inducAssumP},
and since $0 \le \eta \le 1$ and $1-f^u \le 1-f^r$,
it follows that 
\[
SP_P(f,\bigcup_{i=1}^r S'_i) \ge SP_R(f,\bigcup_{i=1}^{r-1} S_i) + SP_R(f, S_r) 
\]
\begin{equation}
\label{eq:ge}
- SP_R(f, S_r) (1-f^r) \overset{(\ref{eq:unionEventsRep})}{=} SP_R(f,\bigcup_{i=1}^r S_i).
\end{equation}
\end{proof}

Note that \PSMARTRED{} preserves the number of shards that \RSMARTRED{} selects per each partition,
and thus according to Equation~\ref{eq:repProws}, 
\PSMARTRED{}'s success probability is equal or greater than \RSMARTRED{}'s.
Nevertheless, although \RSMARTRED{} is optimal for Replication, 
this does not imply that \PSMARTRED{} is optimal for Repartition. 
Note further that Theorem \ref{thm:repartition} holds under the assumption that all partitions' probability distributions are the same.
In practice this assumption does not necessarily hold,
but our experiments show that Repartition is advantageous nevertheless.

\end{sloppypar}

\section{Empirical Study}
\label{sec:emp}
We empirically evaluate our tail-tolerant distributed search using two real-world datasets.
We measure search quality using the recall metric and demonstrate the superiority of 
\RSMARTRED{} over \NORED{} and \RFULLRED{}, 
and the improvement that Repartition suggests over Replication.
Our empirical study confirms our analysis.

\subsection{Methodology}
\label{sec:method}

\paragraph*{Datasets}
\begin{itemize}
	\item Reuters RCV1~\cite{Reuters} consists of news articles spanning a year. 
     We represent each news article by its title and first paragraph. 
     We parse the Reuters datasets using conventional methods, 
     including stop-word removal and stemming~\cite{Manning2008}.
    \item Livej~\cite{Yang2012} was crawled from the Livejournal~\cite{Livejournal} 
          free online community by the Stanford SNAP Project~\cite{SNAP}. 
          In Livejournal, the corpus consists of users who join blogs that reflect their 
          topics of interest. We represent a user by a document 
          where the document's terms correspond to the blogs he
          joined, filtering out users with no topics of interest. 
\end{itemize}
Table \ref{table:datasets} summarizes the dataset statistics after pre-processing.

\begin{table}[hbt]
\begin{center}
\begin{footnotesize}
\begin{tabular}{|c|c|c|}
\hline
       		  & Number of        & Number of        \\ 
              & Documents        & Terms         \\ \hline 
Reuters       & $779\!,913$      & $96\!,513$            \\ 
Livej         & $1\!,147\!,948$  & $664\!,414$           \\ \hline
\end{tabular}\vspace{1em}
\end{footnotesize}
\caption{Dataset statistics after pre-processing.}
\label{table:datasets}
\end{center}
\end{table}

\paragraph*{Experiment setup}
We use Lucene~\cite{Lucene} 4.3.0 search library as our indexing and retrieval infrastructure.
We weight document terms according to Lucene's $\mbox{TF-IDF}$ function, 
where $TF(term)$ is the square root of the term's frequency in the document, 
and $IDF(term) = ln(\frac{N_d}{N_{term}+1})+1$, 
where $N_d$ is the total number of documents,
and $N_{term}$ is the number of items containing the term.
We score documents with Lucene's default similarity function, 
which implements 
a variant of the cosine-based retrieval model.

\paragraph*{Tail-tolerant DiS simulator}
We use the Tarsos-LSH Java implementation~\cite{TarsosLSH} 
of cosine-based LSH for the partitioning.
LSH provides a configuration parameter $k$, which controls the number of shards in the partition.
We partition the data into $n=32$ shards by setting $k$ to $5$.
We simulate the DiS on a single machine by maintaining
a separate index for each shard, as well as for the shard index.
We set $r=3$ in all experiments.

We construct the centralized shard index 
by sampling documents 
from every shard with a configured sampling probability.
Given a query, we first search the shard index and retrieve a result set of top $\gamma$ documents;
we set $\gamma=500$ in our experiments (for all queries).
We then score the shards based on the result set according to CRCS Linear~\cite{Shokouhi2007}:
The score $S(D)$ of shard $D$ is defined as $S(D) \triangleq \sum_{d \in R_D} S(d)$,
where $R_D$ is the subset of the results 
sampled from shard $D$,
and $S(d) = \gamma - j$, where $1 \le j \le \gamma$ is the rank of $d$ in the result set.
We normalize CRCS's scores, $\hat{S}(D)\triangleq \frac{S(D)}{\sum_{D'} S(D')}$, in order to produce the shard success probability distribution that CRCS induces.

We simulate distributed query processing as follows:
Given a query, we retrieve the top $100$ results of each shard.
In order to simulate misses of results, 
we drop the results of each shard with probability $f$.
We union the results of the responsive nodes and omit duplicates (duplicates exist due to redundancy),
which yields a result set $R$ of unique documents.
Since all shards apply the same scoring function, 
we rank the documents in $R$ according to their scores and return the top-scored $100$ documents in $R$.

We examine two shard success probability distributions: 
1) a uniform success probability distribution that we produce using the Random shard selection,
and 2) a skewed success probability distribution that we produce using CRCS with sampling probability of $0.4$\footnote{In order to produce a highly skewed success probability using CRCS,
for the purpose of demonstration, 
we use an extremely high sampling probability of $0.4$.}. 
Figure~\ref{fig:emp_spDist} illustrates the 
average estimated success probability of 
the five top-scored shards for the LiveJ and the Reuters datasets
produced as follows:
For each query in the evaluation set, 
we estimate the success probabilities of the query's top five shards.
We then average each of these five success probabilities over all queries.
As the figure illustrates,
the Random shard selection induces a uniform success probability distribution,
which is identical for both datasets (Figure~\ref{fig:emp_spDist}\subref{subfig:lj_uni_sp}).
CRCS induces a skewed success probability distribution (Figures~\ref{fig:emp_spDist}\subref{subfig:lj_skew_sp}
and~\ref{fig:emp_spDist}\subref{subfig:rcv_skew_sp}),
and in particular the most skewed one in LiveJ (Figure~\ref{fig:emp_spDist}\subref{subfig:lj_skew_sp}).

\begin{figure*}[tb]
	\centering     
    \subfloat[Random]{\label{subfig:lj_uni_sp}\includegraphics[scale=0.4]{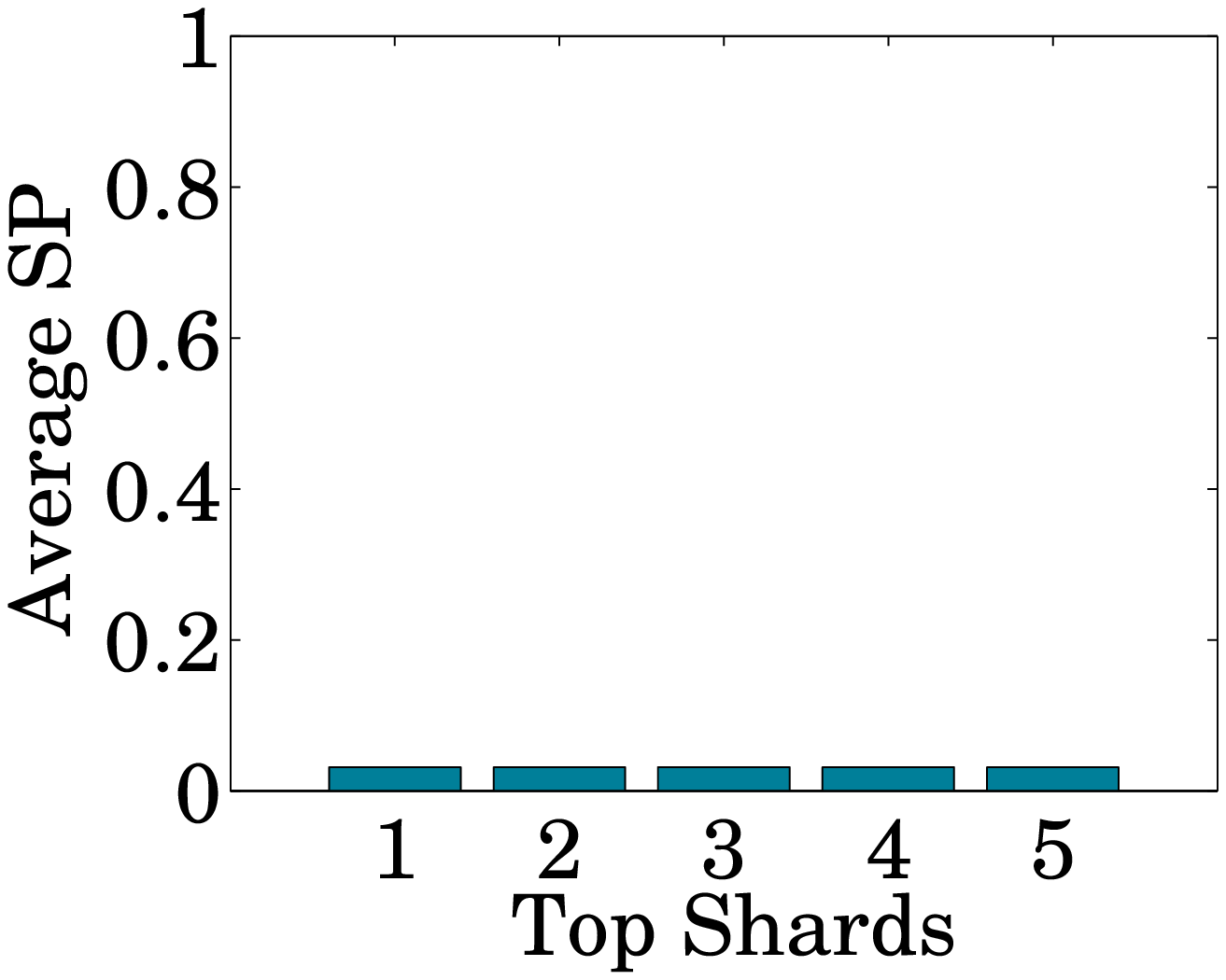}}\hspace{0.5em} 
    \subfloat[LiveJ CRCS]{\label{subfig:lj_skew_sp}\includegraphics[scale=0.4]{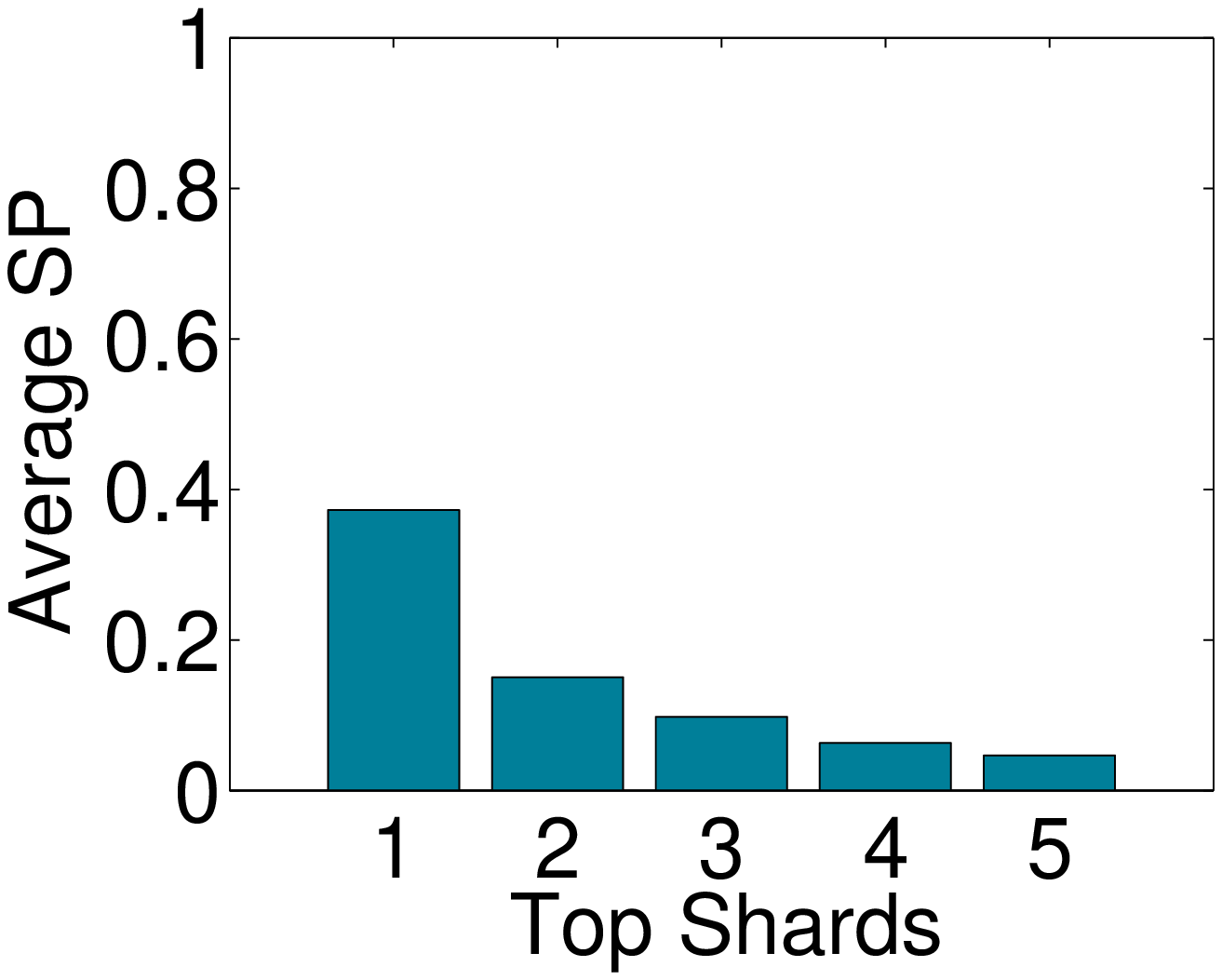}}\hspace{0.5em}      
    \subfloat[Reuters CRCS]{\label{subfig:rcv_skew_sp}\includegraphics[scale=0.4]{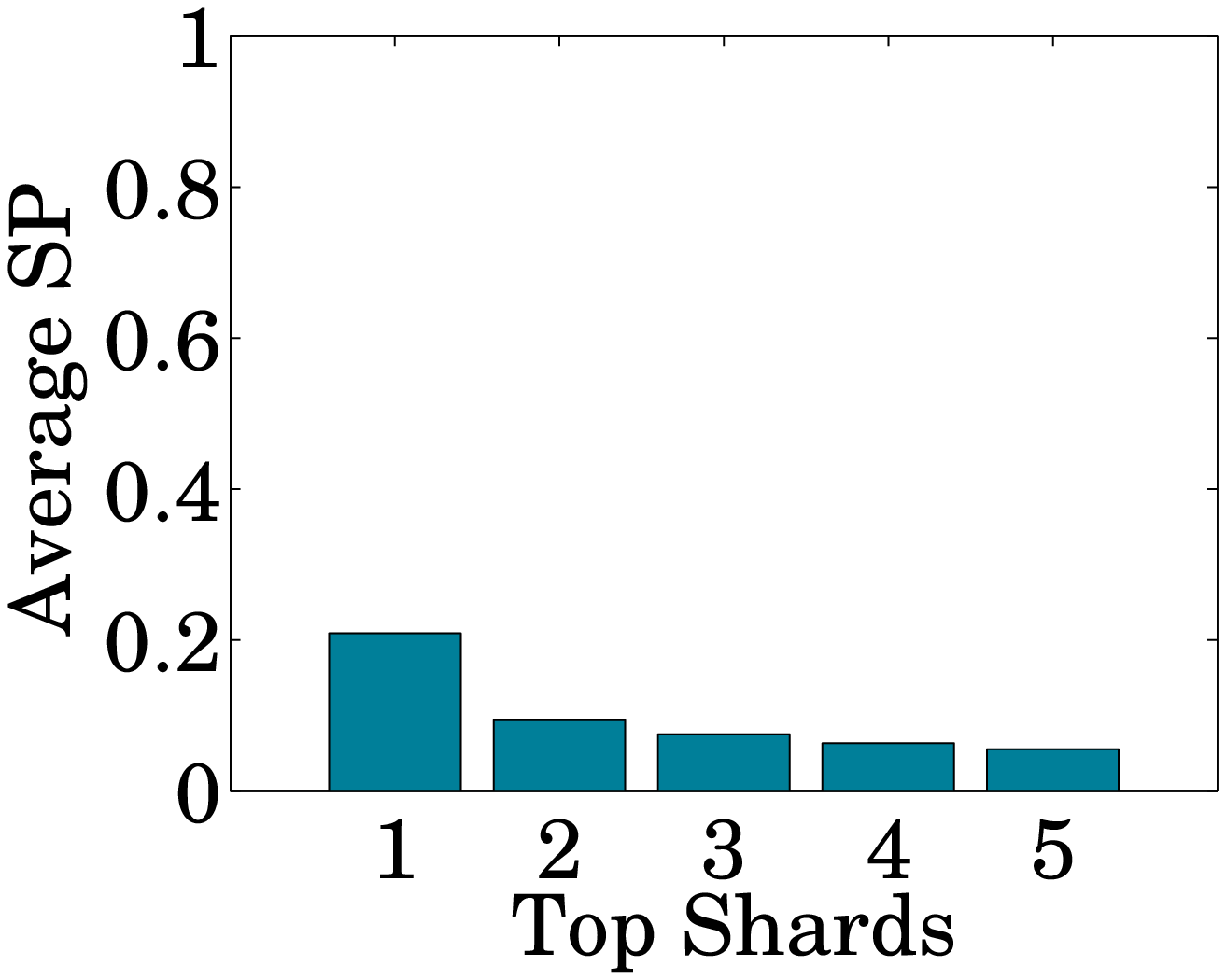}}
    \vspace{1em}
    \caption{Random and CRCS success probability distributions over the Reuters and LiveJ datasets.}    
	\label{fig:emp_spDist}
    
    \subfloat[Random]{\label{subfig:lj_Random}\includegraphics[scale=0.4]{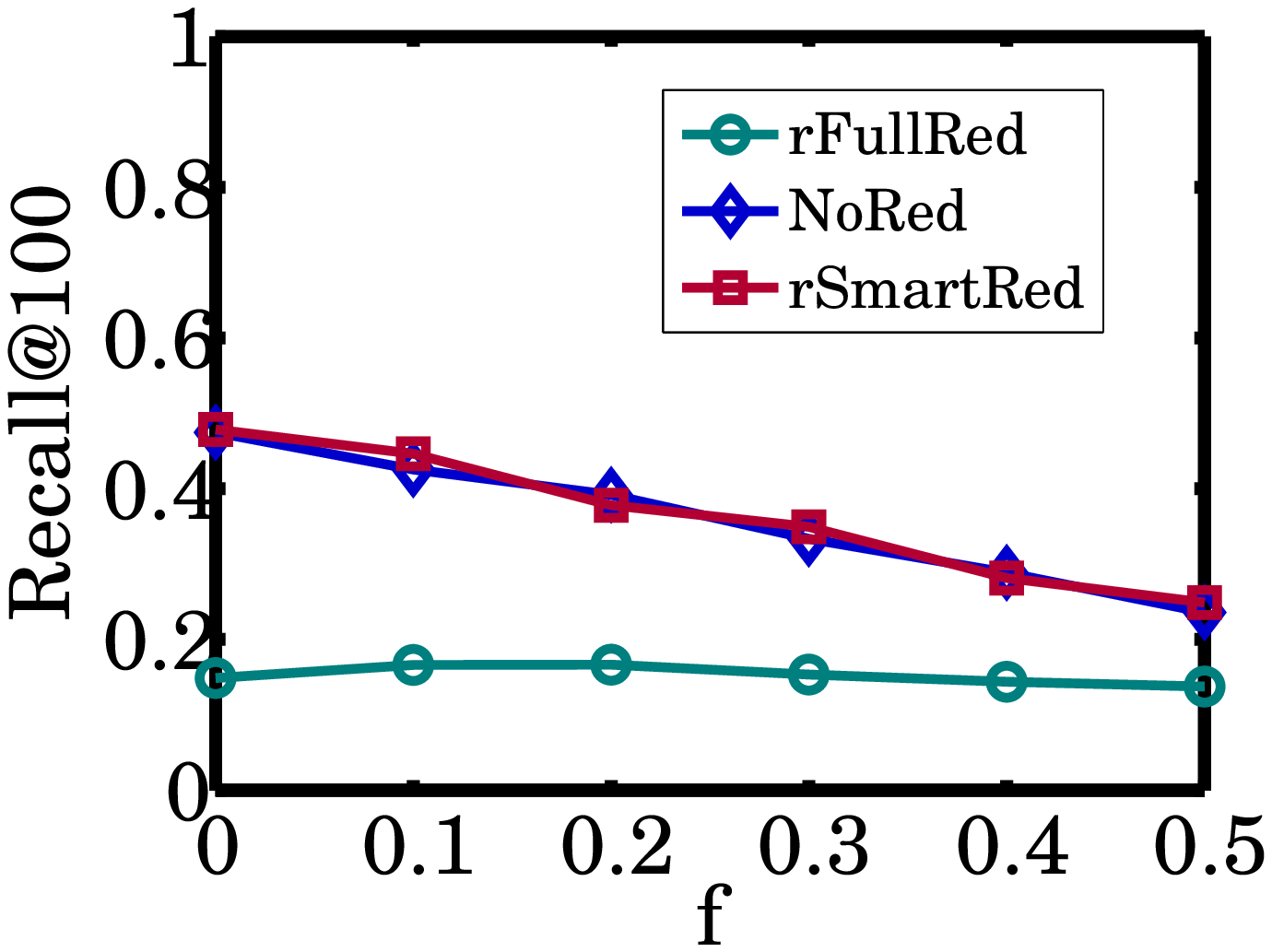}}\hspace{0.5em}   
    \subfloat[LiveJ CRCS]{\label{subfig:lj_ReDDE}\includegraphics[scale=0.4]{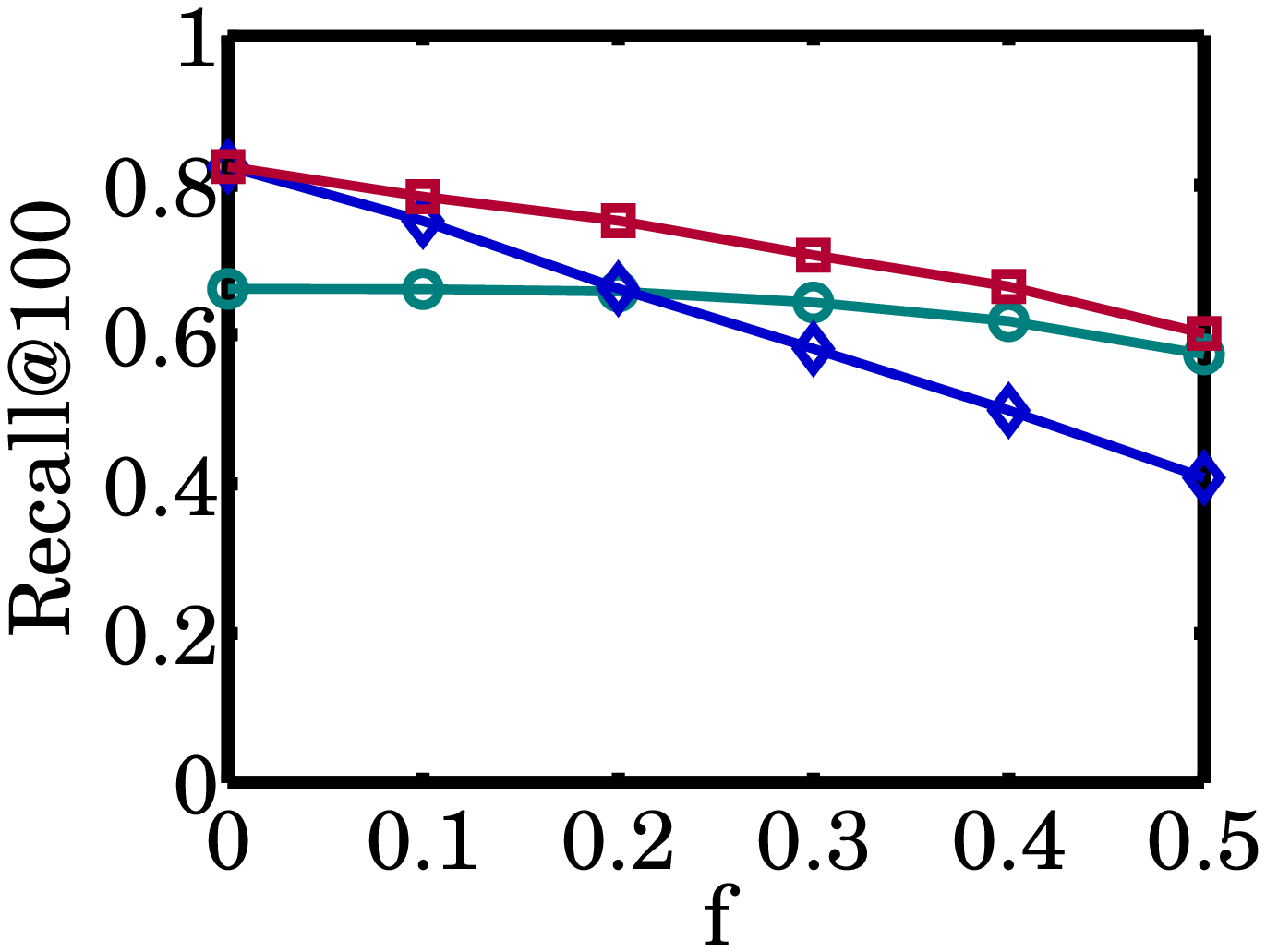}}\hspace{0.5em}    
    \subfloat[Reuters CRCS]{\label{subfig:rcv_ReDDE}\includegraphics[scale=0.4]{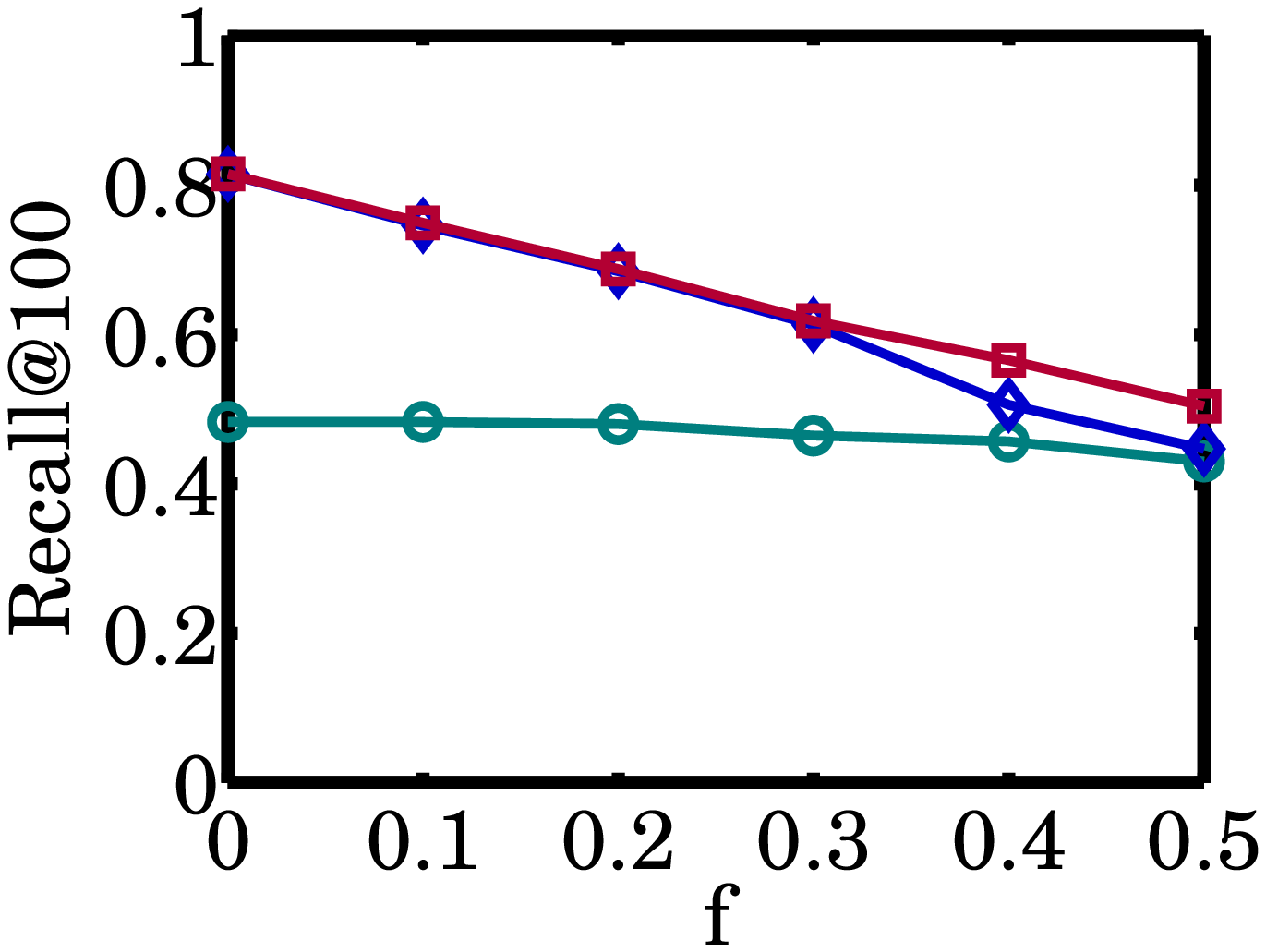}}
    \vspace{1em}
    \caption{Recall@100 for the three shard selection schemes for Replication as a function of miss probability $f$.
     \RSMARTRED{} outperforms both \RFULLRED{} and \NORED{} in all scenarios.}  
	\label{fig:emp_lj}
\end{figure*}

\COMMENT {
\paragraph*{CORI}
For every shard $D_i$ in the LSH partition,
we construct a corresponding centroid vector of the shard's features.
In order to avoid large vectors, we prune the shard vector features: 
we sort the vector's features according to their frequency in the shard,
and leave a fraction of $\beta$ most frequent features, pruning the rest.
For a given dataset, We select $\beta$ such that the number of features in CORI's shard index approximately equals 
the number of features in the shard index.
Thus, in Reuters, we set $\beta=0.25$ and in Livej we set $\beta=0.28$.
After pruning the shard vectors, we index them into the shard index.
We use Lucene for weighting the term features according to $TF$ and $ICF$.
Given a query, we execute it against the Lucene's shard index and retrieve the shards whose centroid vectors
are most similar to the query.
}

\paragraph*{Evaluation}
We evaluate search quality by measuring the average Recall@100 over an evaluation set of queries
(see Section \ref{sec:metric}).
For Reuters, we use $200$ Trec topics~\cite{ReutersTopics} 
as our evaluation set.
For LiveJ, we construct an evaluation set by randomly sampling $1\!,000$ documents from the dataset.
We confirm statistical significance using 
paired-ttest with $5\%$ significance level.

\subsection{Selection schemes for Replication}
\label{sec:emp_replication}
We study \RFULLRED{}, \NORED{}, and \RSMARTRED{} selection schemes for Replication
and show the superiority of \RSMARTRED{} over the other two.

\paragraph*{Effect of miss probability and success probability distribution}
We evaluate recall as a function of $f$, $0 \le f \le 0.5$
(we discuss here relatively high $f$ values that are not necessarily realistic
for the purpose of the demonstration; we later zoom-in on lower $f$ values).
We fix $t=5$, i.e., we select $tr=15$ shards, which are about half of the number of shards in one partition ($32$).
Figure~\ref{fig:emp_lj} presents our results 
for the Random and the CRCS-based success probability distributions
(Random's results are very similar for both datasets,
and so, we present here LiveJ's results only).

For both success probability distributions and both datasets, 
\RFULLRED{} achieves relatively stable recall for all miss probabilities.
This is since \RFULLRED{} uses the maximal number of replicas ($r$) possible for all shards it selects,
which makes the content it selects available regardless of $f$.
However, the recall \RFULLRED{} achieves is significantly lower for low miss probabilities
compared to \NORED{}.
This is since searching a large number of replicas is wasteful when misses are infrequent;
it is more beneficial to select additional shards by decreasing the number of shard replicas.
At the other extreme, \NORED{} selects a single replica of each shard.
It achieves higher recall compared to \RFULLRED{} when miss probability is low,
as it searches more distinct shards.
However, when miss probability increases, \NORED{}'s recall decreases.
This tradeoff is most pronounced in LiveJ when the success probability distribution is highly skewed 
(Figure~\ref{fig:emp_lj}\subref{subfig:lj_ReDDE}),
where \NORED{}'s recall drops below the recall of \RFULLRED{} for $f$ 
values that exceed $0.2$.
This is since in the case of a skewed success probability distribution,
the responsiveness of top-scored shards is crucial,
hence contacting their replicas is beneficial.
When the success probability distribution is uniform 
(Figurs~\ref{fig:emp_lj}\subref{subfig:lj_Random}),
\NORED{}'s recall becomes close to that of \RFULLRED{} for $f=0.5$.
This is since in this case, when a shard fails to respond,
contacting one of its replicas (\RFULLRED{}) or contacting another shard in the partition (\NORED{}) contributes 
similarly to the success probability.

As expected, for both distributions, for all values of $f$,
\RSMARTRED{}'s recall is at least as good as those of \RFULLRED{} and \NORED{}.
When the success probability distribution is close to uniform (Figure~\ref{fig:emp_lj}\subref{subfig:lj_Random}), 
\RSMARTRED{} and \NORED{} behave similarly,
since not much is gained from redundancy.
But when the success probability distribution is skewed (Figures~\ref{fig:emp_lj}\subref{subfig:lj_ReDDE} and~\ref{fig:emp_lj}\subref{subfig:rcv_ReDDE}),
as is common for many queries,
\RSMARTRED{} outperforms both \RFULLRED{} and \NORED{} by adjusting its selection to the miss probability.
For example, as demonstrated in Figure~\ref{fig:emp_lj}\subref{subfig:lj_ReDDE},
\RSMARTRED{} achieves a statistically significant higher recall then \RFULLRED{} 
for $f < 0.2$ and a statistically significant higher recall than
\NORED{} for $f > 0.2$.

\COMMENT {
\begin{figure*}[tb]
	\centering     
    \subfloat[Random]{\label{subfig:lj_Random}\includegraphics[scale=0.45]{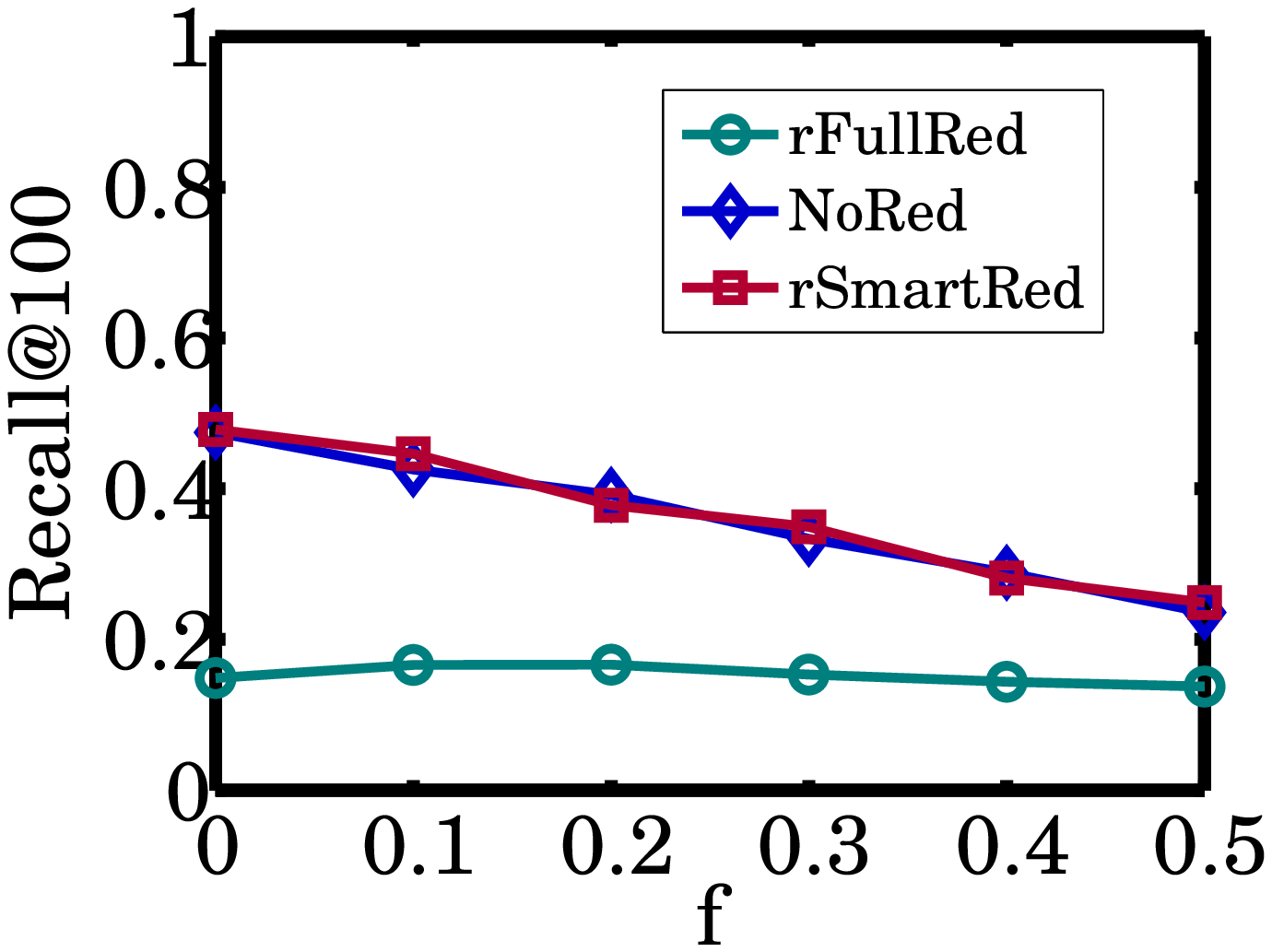}}\hspace{0.5em}     
    \subfloat[LiveJ CRCS]{\label{subfig:lj_ReDDE}\includegraphics[scale=0.45]{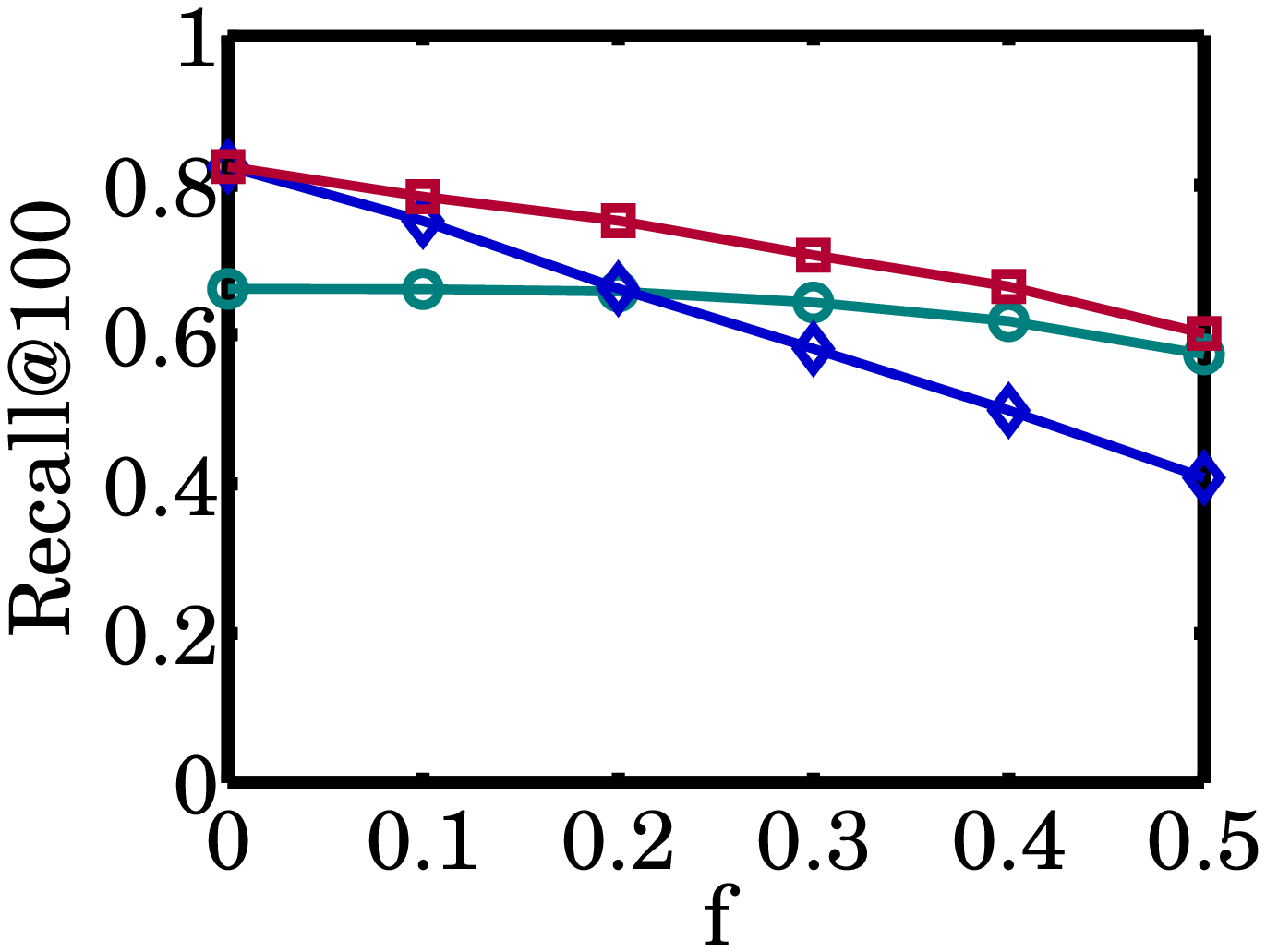}}\hspace{0.5em}    
    \subfloat[Reuters CRCS]{\label{subfig:rcv_ReDDE}\includegraphics[scale=0.45]{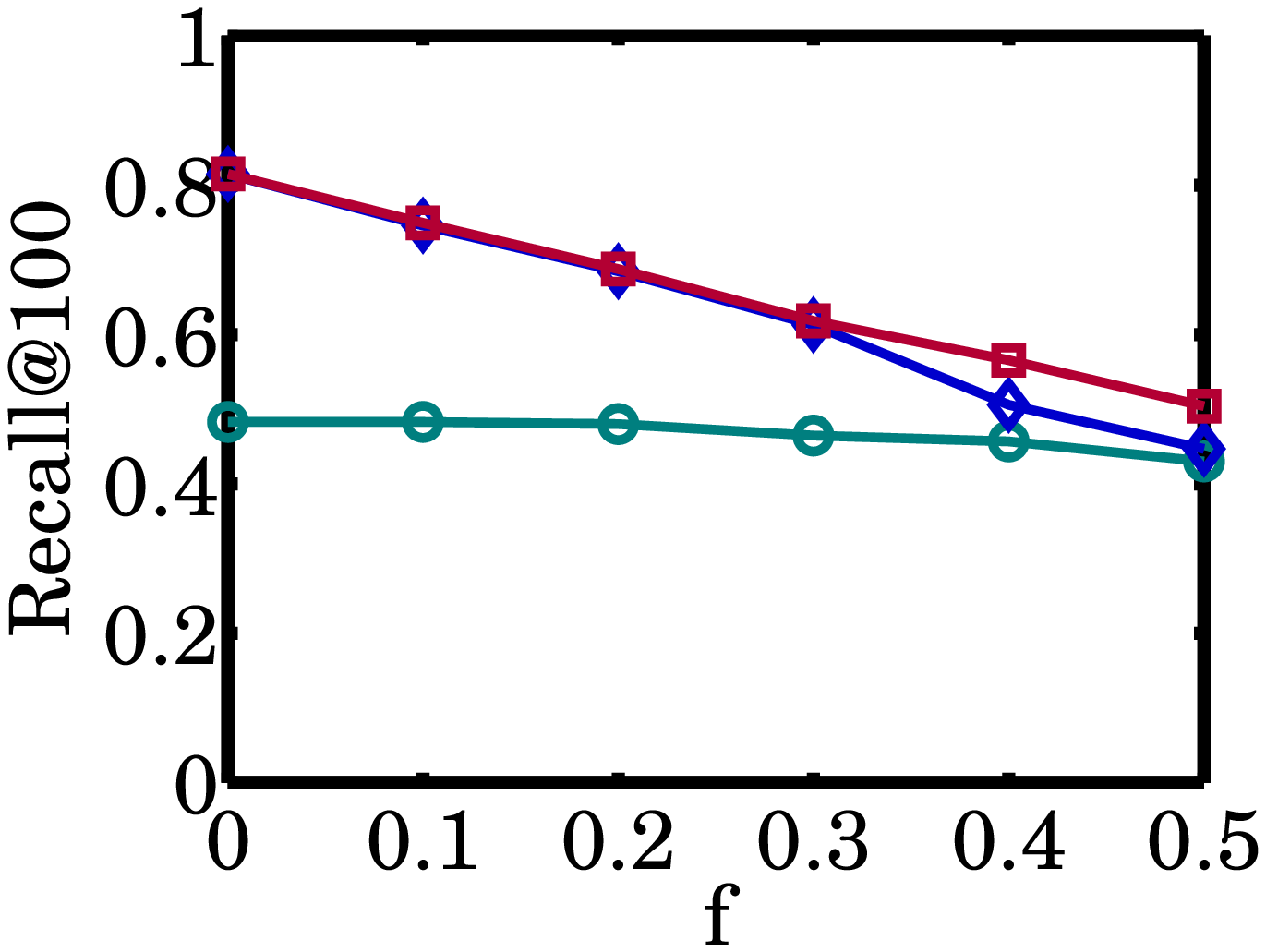}}
    \vspace{1em}
    \caption{Recall@100 for the three shard selection schemes for Replication as a function of miss probability $f$.
     \RSMARTRED{} outperforms both \RFULLRED{} and \NORED{} in all scenarios.}  
  
	\label{fig:emp_lj}
\end{figure*}
}

We next zoom in on smaller $f$ values, $0 \le f \le 0.2$.
We examine three CRCS-based (non-uniform) success probability distributions for LiveJ 
by considering the following three sets of queries:
\begin{itemize}
	\item \emph{\skewedSet{}} consists of all queries in LiveJ's evaluation set. 
	\item \emph{\moreSkewedSet{}} consists of queries in LiveJ's evaluation set
	for which the success probability of the top shard is greater than $0.5$; 
    26.3\% of queries are in this category.
    \item \emph{\mostSkewedSet{}} consists of queries in LiveJ's evaluation set
	for which the success probability of the top shard is greater than $0.8$;
    only 0.092\% of queries are in this category.
\end{itemize}
Figure~\ref{fig:emp_spDistSkewed} illustrates the 
average estimated success probability of 
the five top-scored shards for each of the query sets.

\begin{figure*}[tb]
	\centering     
    \subfloat[\skewedSet{} query set]{\label{subfig:lj_sp_whole}\includegraphics[scale=0.4]{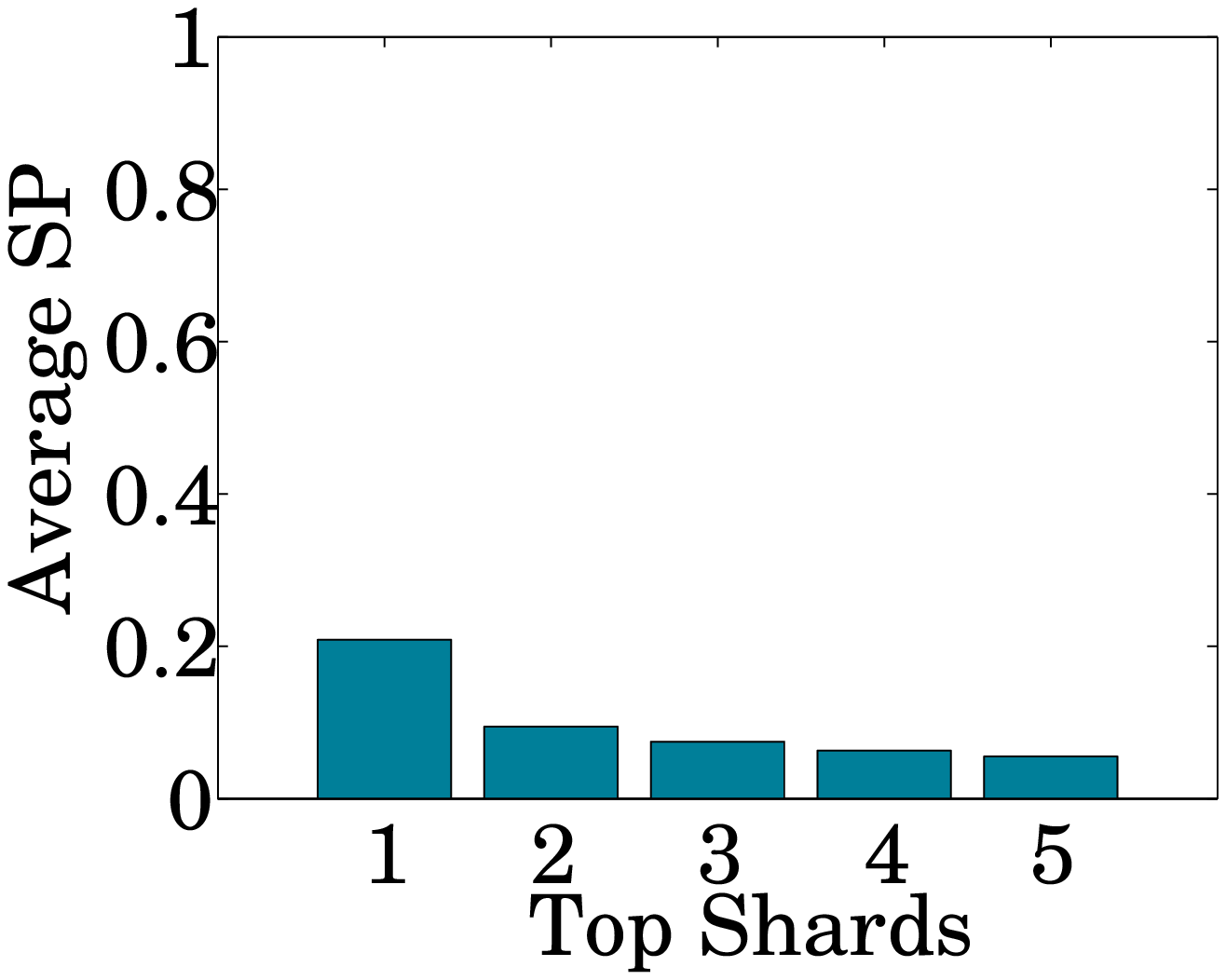}}\hspace{0.5em}  
    \subfloat[\moreSkewedSet{} query set]{\label{subfig:lj_sp_05}\includegraphics[scale=0.4]{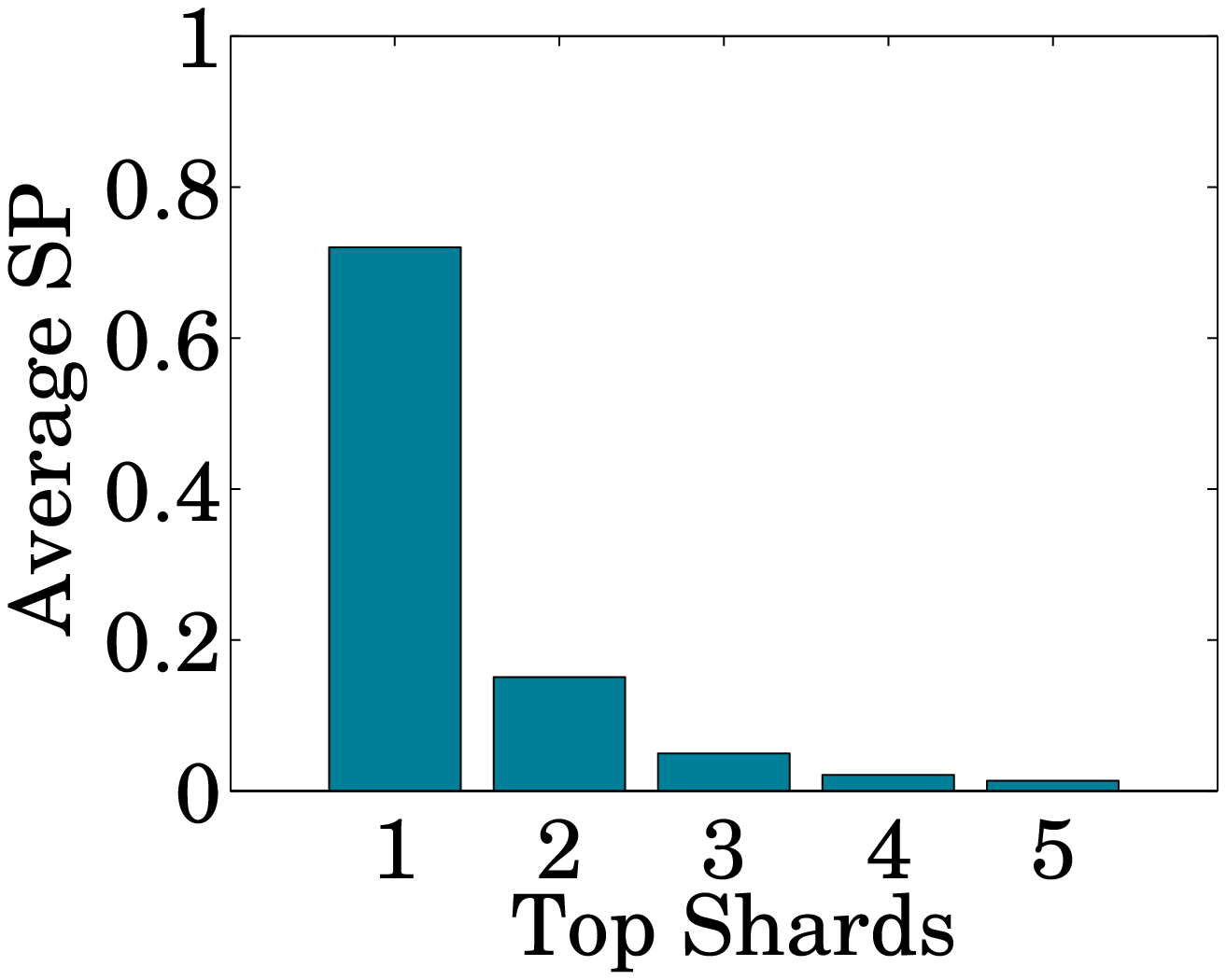}}\hspace{0.5em}    
    \subfloat[\mostSkewedSet{} query set]{\label{subfig:lj_sp_08}\includegraphics[scale=0.4]{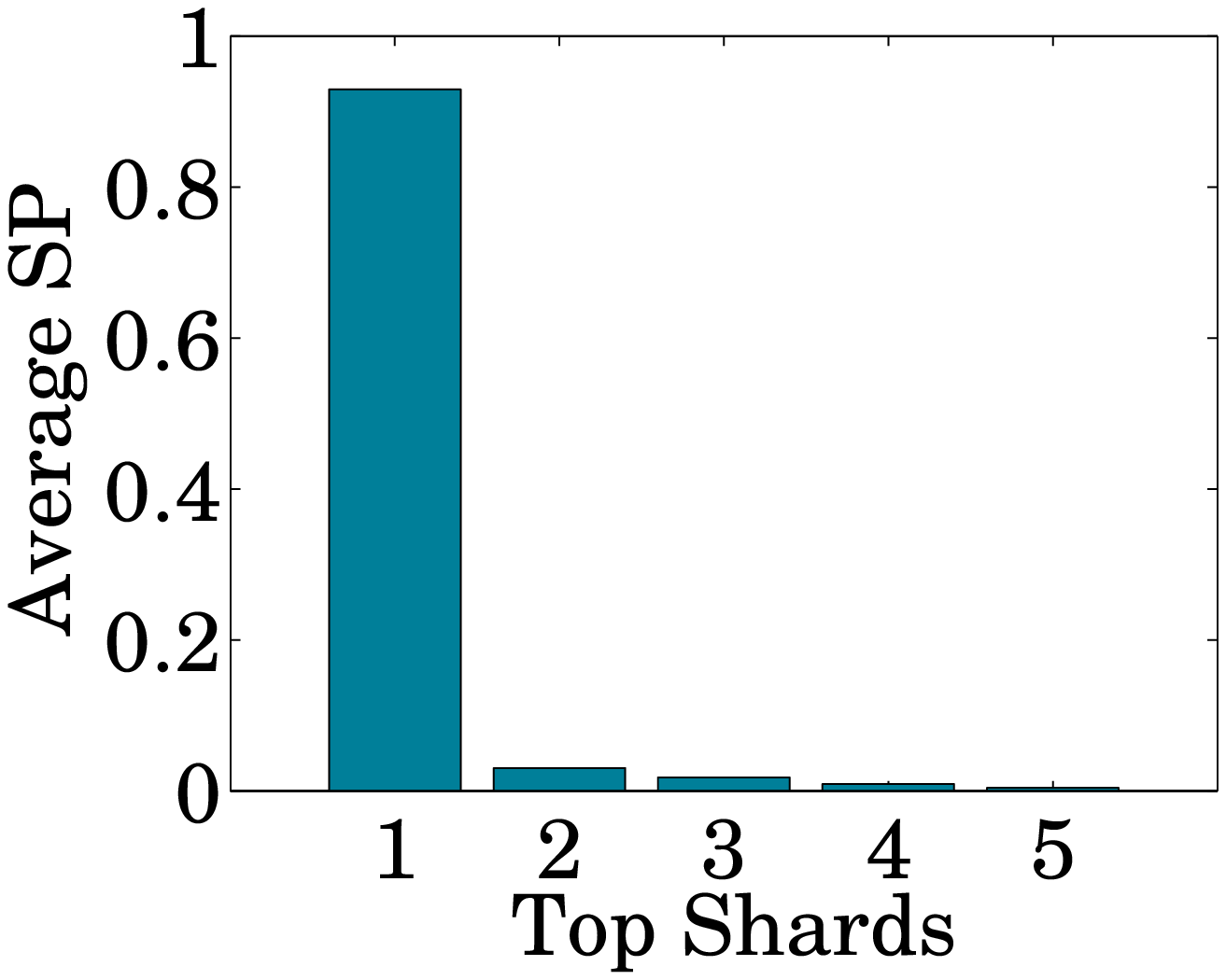}}
    \vspace{1em}
    \caption{Three non-uniform success probability distributions for LiveJ with CRCS-based shard selection, which correspond to three different query sets.}    
	\label{fig:emp_spDistSkewed}
    
   \subfloat[\skewedSet{} query set]{\label{subfig:skew_allExp}\includegraphics[scale=0.4]{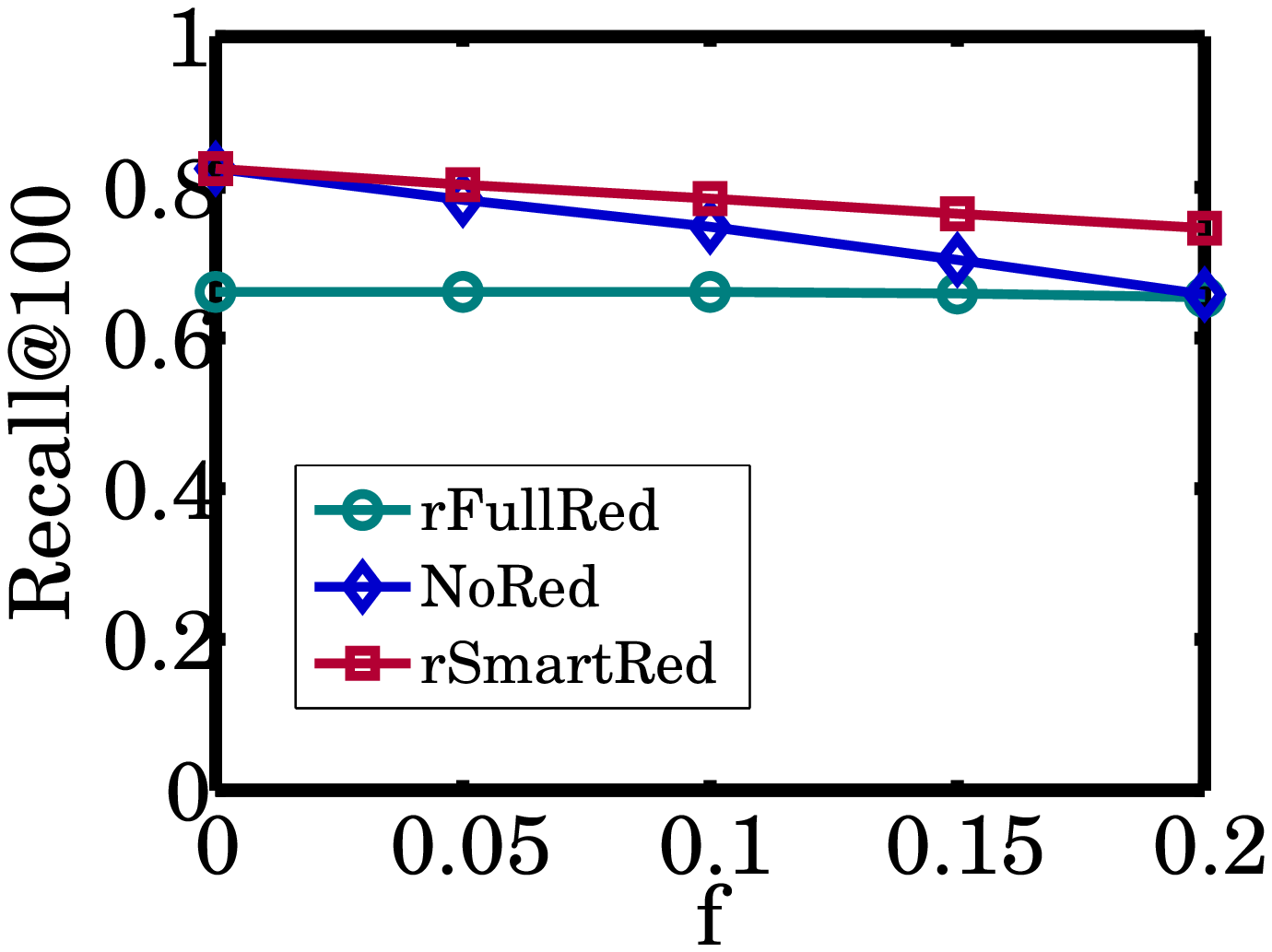}}\hspace{0.5em}     
    \subfloat[\moreSkewedSet{} query set]{\label{subfig:skew_top0_5}\includegraphics[scale=0.4]{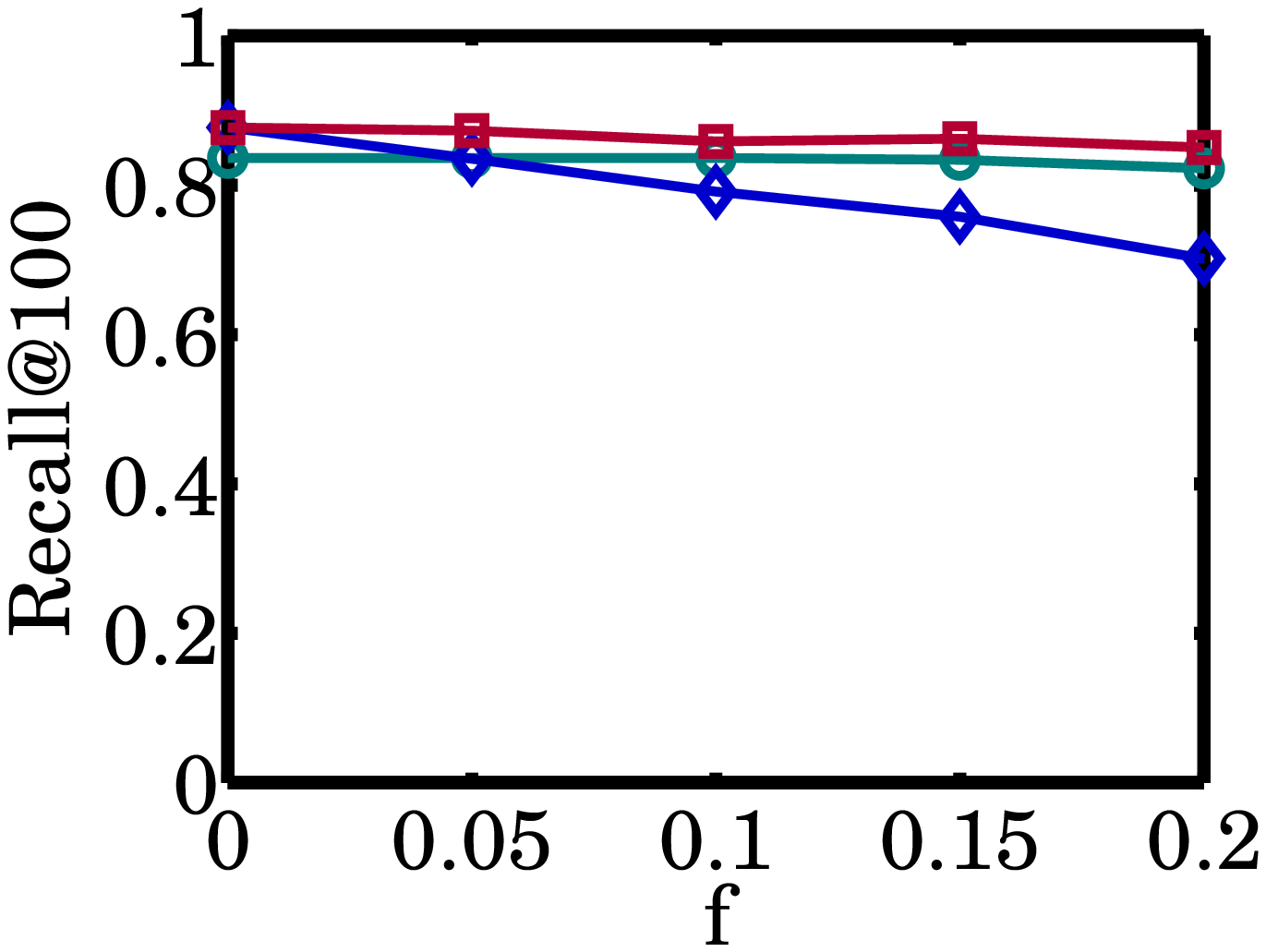}}\hspace{0.5em}    
    \subfloat[\mostSkewedSet{} query set]{\label{subfig:skew_top0_8}\includegraphics[scale=0.4]{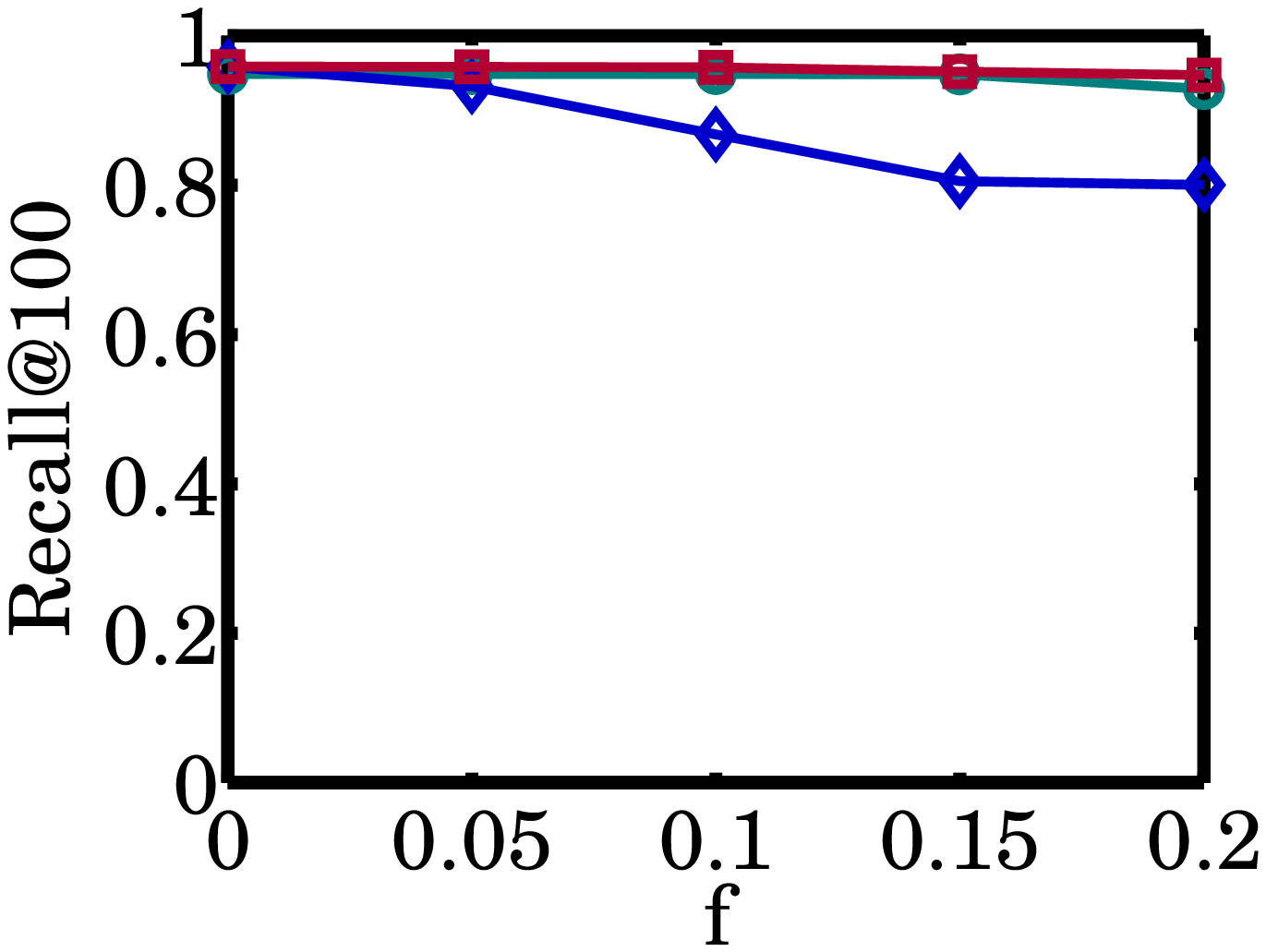}}   
    \vspace{1em}
    \caption{LiveJ Recall@100 for the three shard selection schemes for Replication as a function of miss probability $f$ for three different query sets, each inducing a different success probability distribution. \RSMARTRED{} outperforms both \RFULLRED{} and \NORED{} in all scenarios.}  
	\label{fig:skew_lj}
    
        \subfloat[LiveJ Random]{\label{subfig:lj_Random_t}\includegraphics[scale=0.4]{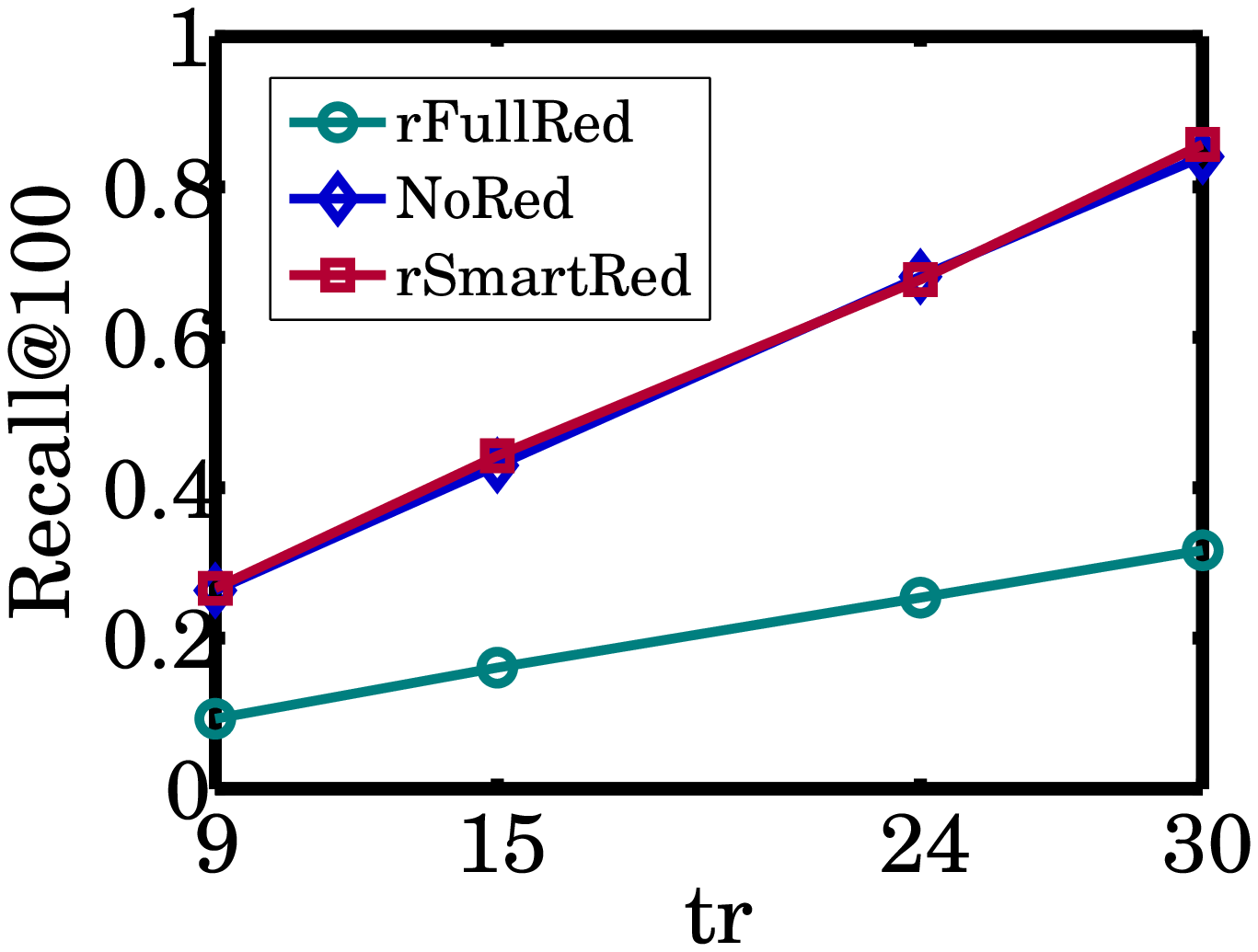}}\hspace{1em} 
    \subfloat[LiveJ CRCS]{\label{subfig:lj_ReDDE_t}\includegraphics[scale=0.4]{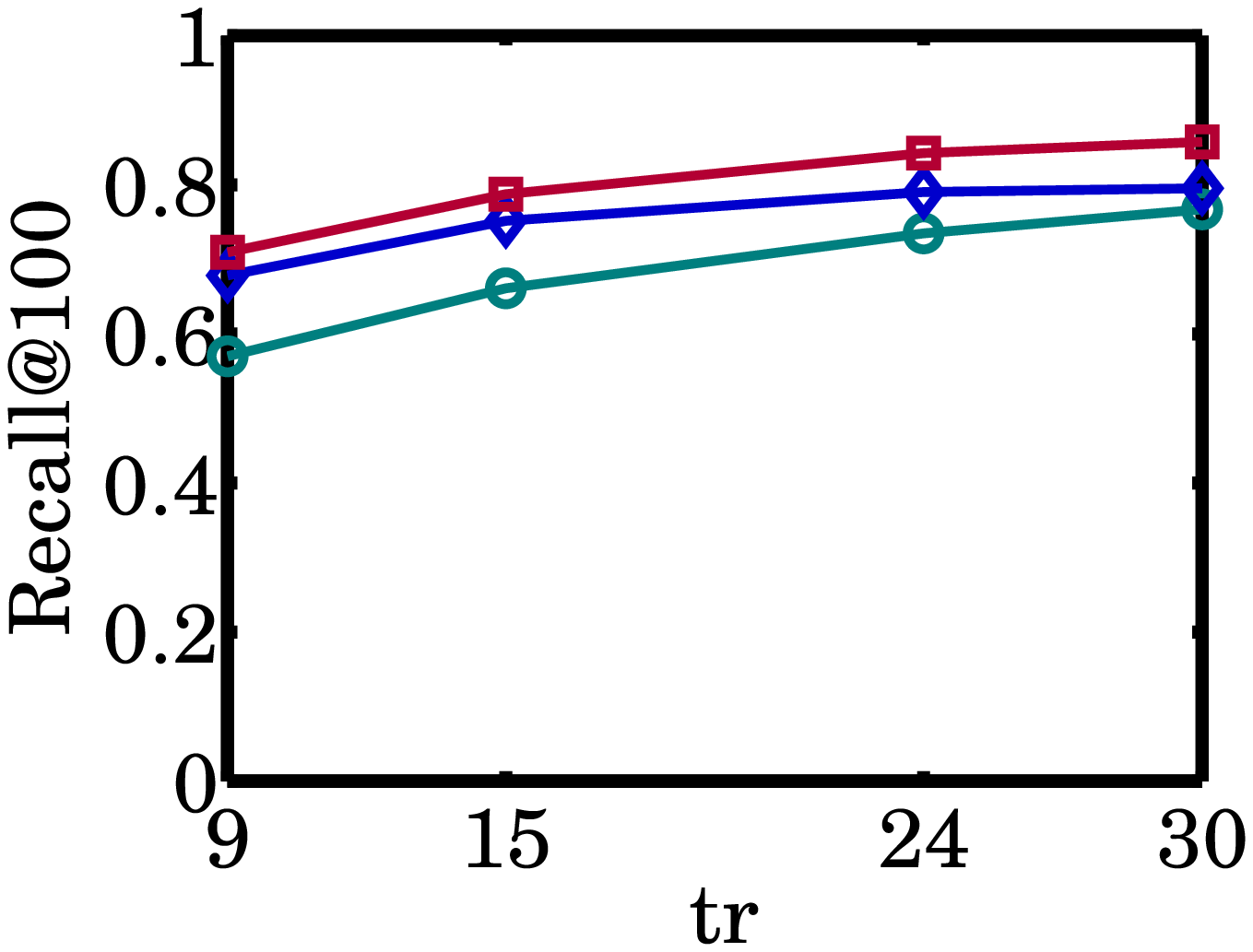}} 
    \vspace{1em}
    \caption{LiveJ Recall@100 for the three shard selection schemes for Replication as a function of the number of selected shards ($tr$).
    \RSMARTRED{} outperforms both \RFULLRED{} and \NORED{} in all scenarios.}  
	\label{fig:emp_byNShards}
    
    \subfloat[By miss probability]{\label{subfig:both_f}\includegraphics[scale=0.4]{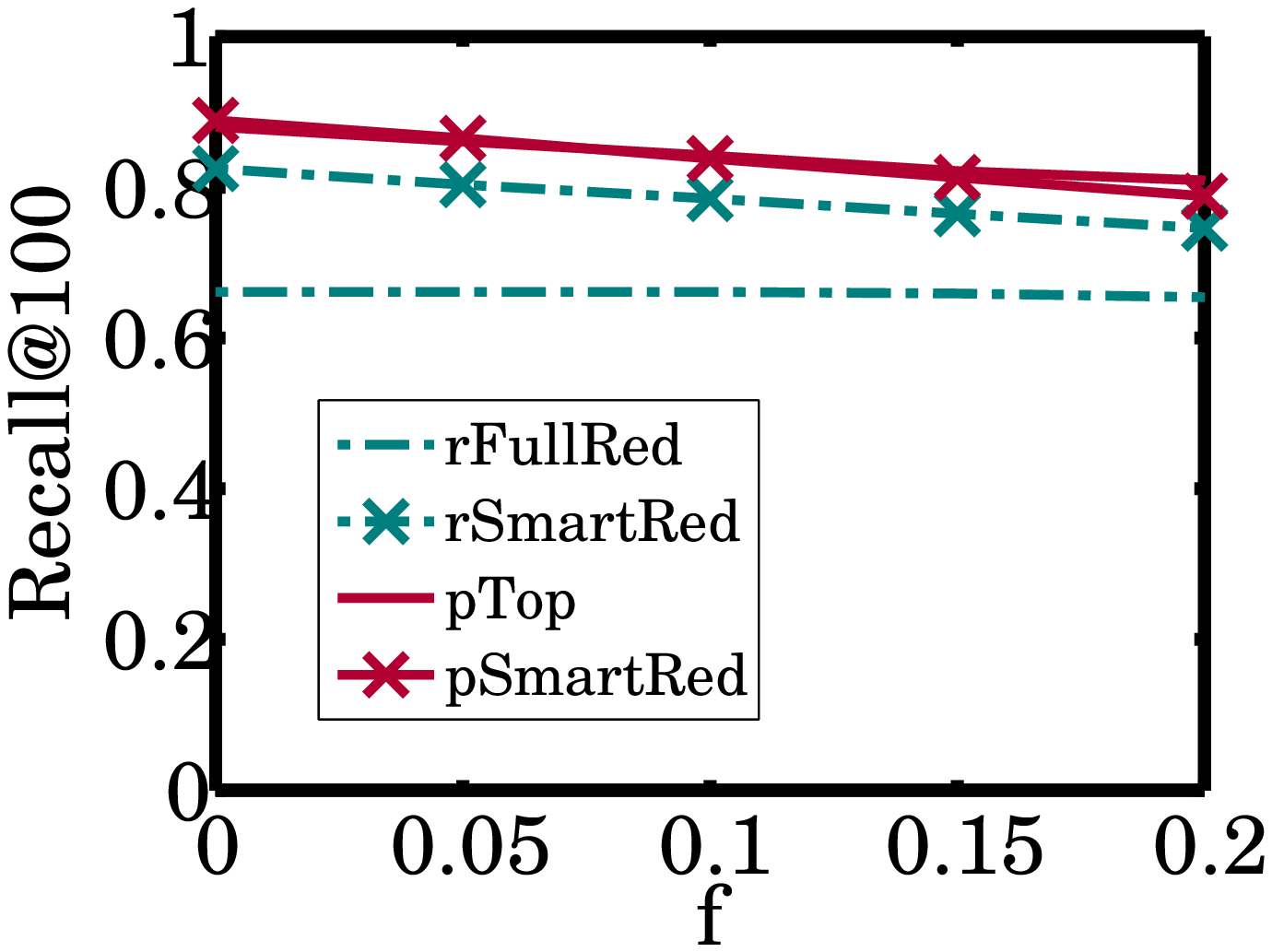}}\hspace{1em} 
    \subfloat[By number of selected shards]{\label{subfig:both_t}\includegraphics[scale=0.4]{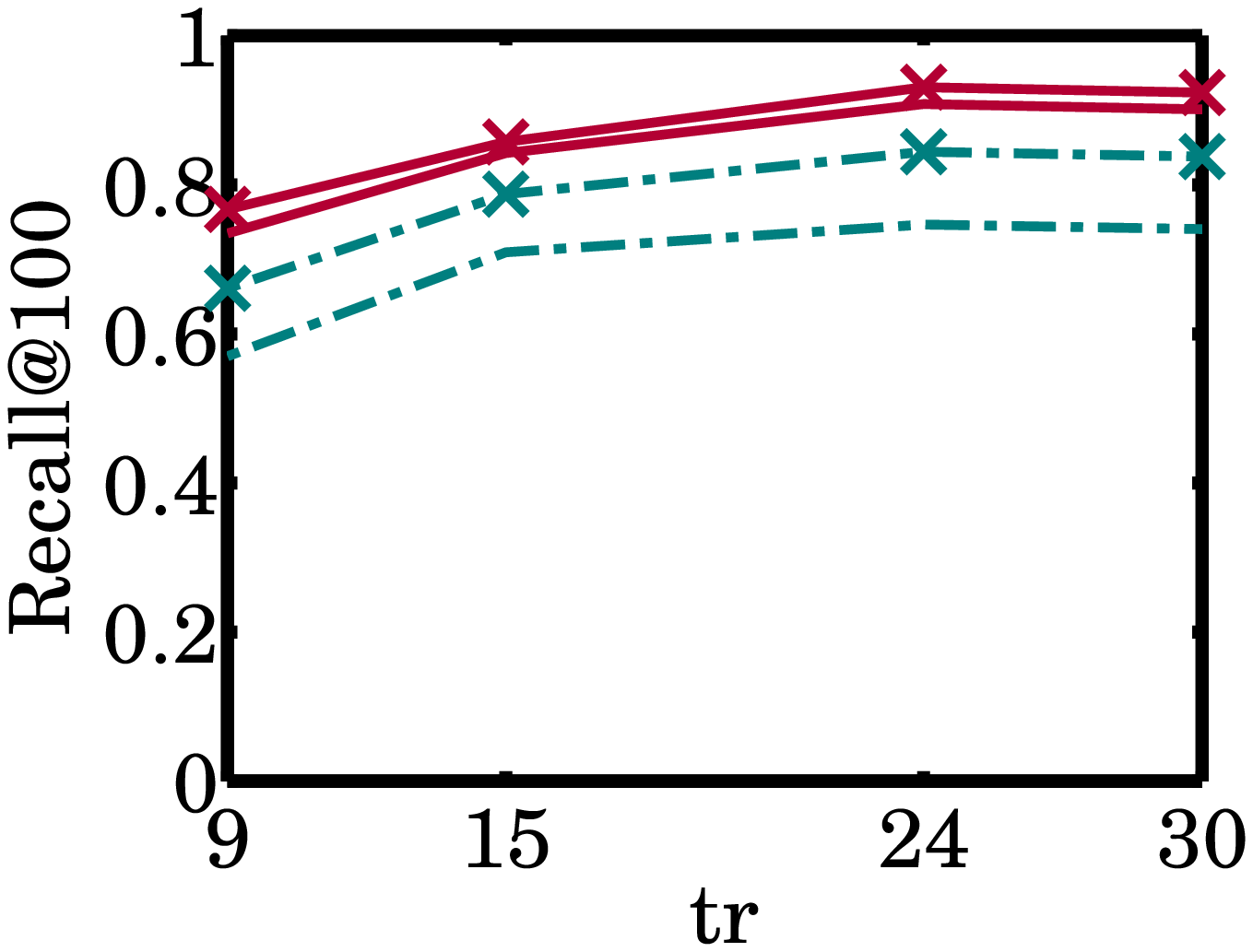}}
    \vspace{1em}
    \caption{LiveJ Recall@100 with Replication and Repartition for a skewed success probability distribution. Repartition outperforms Replication.}
    \label{fig:emp_RvsP}   
\end{figure*}

We compare the recall of the three selection schemes over each of the three
query sets and present our results in Figure~\ref{fig:skew_lj}.
Figure~\ref{fig:skew_lj}\subref{subfig:skew_allExp} depicts recall measured over the \skewedSet{} query set.
As we have previously seen, 
\NORED{} outperforms \RFULLRED{} for $0 \le f \le 0.2$,
however, \NORED{}'s recall decreases as $f$ increases until it reaches \RFULLRED{}'s recall for $f=0.2$.
As \RSMARTRED{} is optimal, it outperforms both.

Next, in Figure~\ref{fig:skew_lj}\subref{subfig:skew_top0_5},
we depict recall for the \moreSkewedSet{} query set,
which has a more skewed success probability distribution 
than \skewedSet{}'s.
Here too, \NORED{}'s recall is higher than \RFULLRED{}'s for low miss probabilities,
but it drops below \RFULLRED{} at a smaller $f$ value ($0.05$).
This is since here, the responsiveness of the top shards is more crucial due to their higher success probability.
Hence, Replication becomes valuable for lower miss probability values than in the previous case.
As before, \RSMARTRED{} selects the number of replicas in an optimal manner and so outperforms
both \RFULLRED{} and \NORED{}.

Finally, Figure~\ref{fig:skew_lj}\subref{subfig:skew_top0_8} examines the recall over the \mostSkewedSet{} query set, 
which has the most skewed success probability distribution.
In this extreme case, the average success probability of the top shard is $0.92$,
hence searching a single shard -- the top one -- is crucial.
As both \RFULLRED{} and \NORED{} select the top shard 
due to its high success probability (Observation~\ref{obs:topPerPartition}),
they both achieve high recall ($0.96$) when misses are infrequent ($f \le 0.05$).
When the miss probability increases ($f > 0.05$), \NORED{}'s recall drops below \RFULLRED{}'s recall as \NORED{} does not employ redundancy. 
For all $f$ values that we examined, \RSMARTRED{} is optimal.

\COMMENT {
\begin{figure*}[tb]
	\centering     
    \subfloat[\skewedSet{} query set]{\label{subfig:skew_allExp}\includegraphics[scale=0.4]{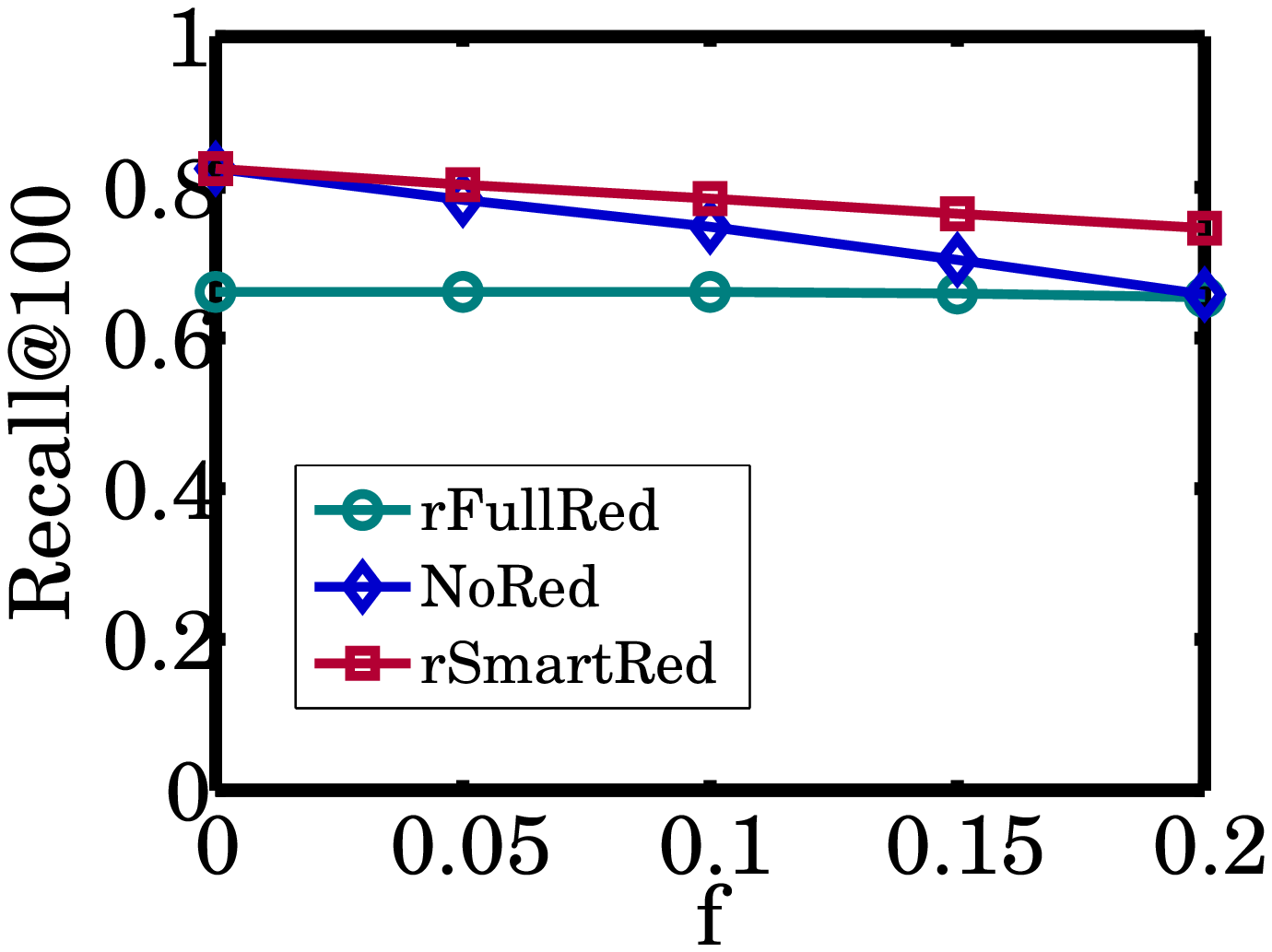}}\hspace{2em}     
    \subfloat[\moreSkewedSet{} query set]{\label{subfig:skew_top0_5}\includegraphics[scale=0.4]{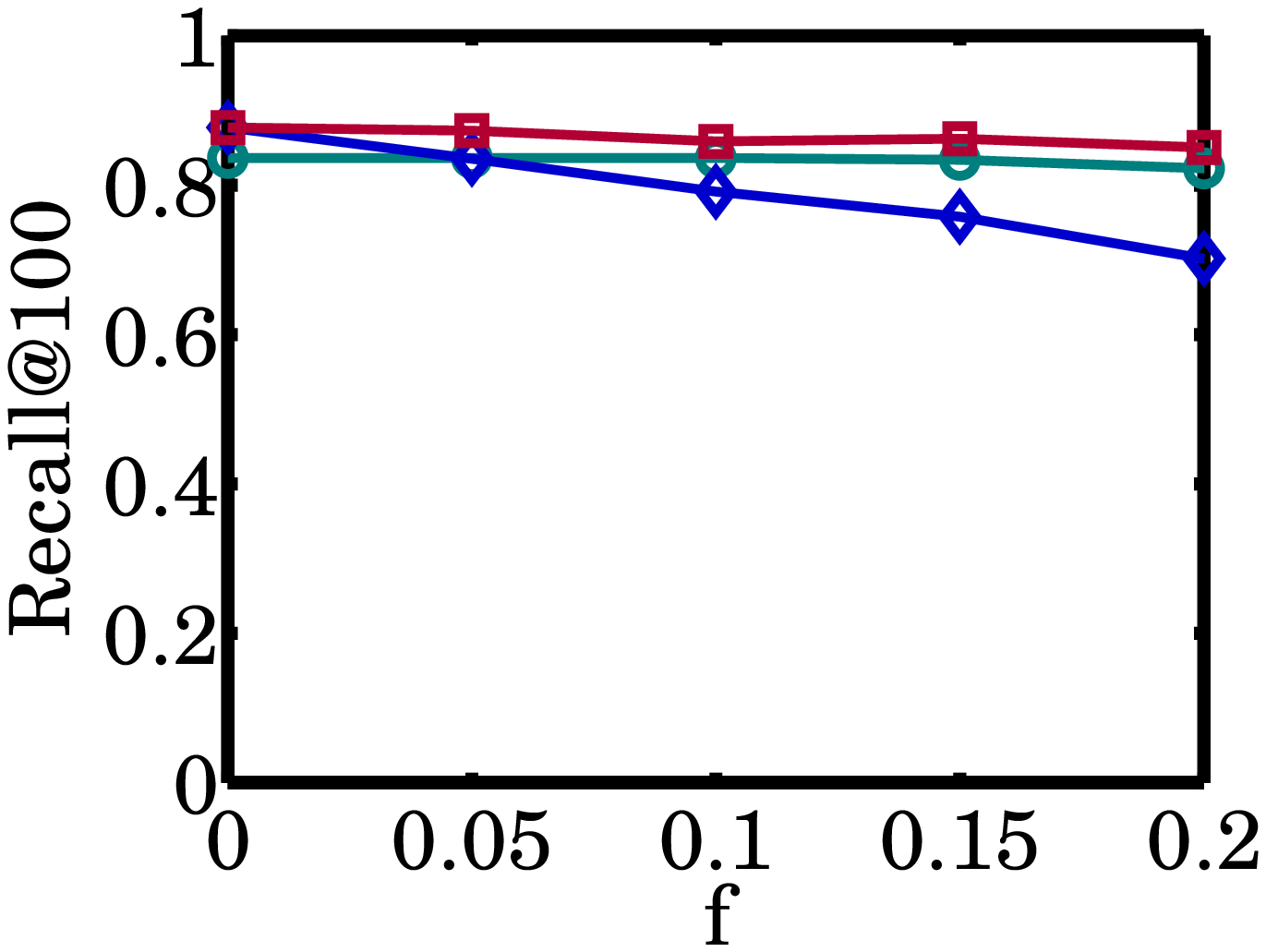}}\hspace{2em}    
    \subfloat[\mostSkewedSet{} query set]{\label{subfig:skew_top0_8}\includegraphics[scale=0.4]{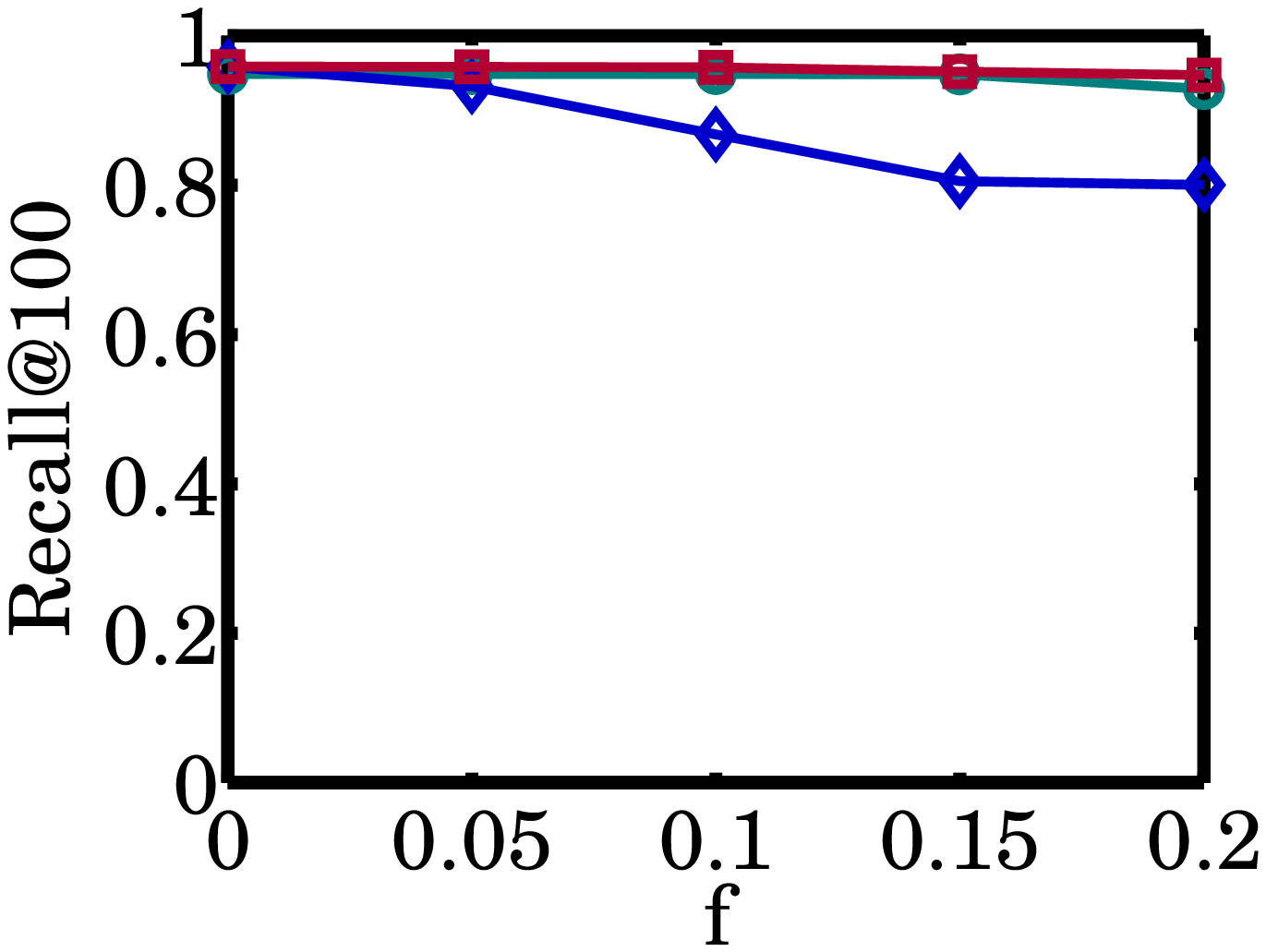}}   
    \vspace{1em}
    \caption{LiveJ Recall@100 for the three shard selection schemes for Replication as a function of miss probability $f$. We consider three different query sets, each inducing a different success probability distribution. \RSMARTRED{} outperforms both \RFULLRED{} and \NORED{} in all scenarios.}  
	\label{fig:skew_lj}
\end{figure*}
}

\paragraph*{Effect of the number of selected shards}
We wrap up the study of Replication by varying $tr$, i.e., the number of selected shards.
We experiment with $t \in \left\{3,5,8,10\right\}$ values.
As we fix $r=3$, this yields up to $30$ shards,
which is almost the number of shards in the partition: $n=32$.
We fix $f=0.1$.
Figure~\ref{fig:emp_byNShards} depicts our results for the LiveJ dataset
(we omit Reuters results, which do not provide additional insight).
For all selection schemes, 
the recall 
increases with
the number of selected shards $tr$, as expected.
Second, for all $tr$ values that we examine, \RSMARTRED{}'s recall is equal to or greater than the recall of 
\NORED{} and \RFULLRED{}, which confirms our theory. 

Searching shard replicas is useless in case of a uniform success probability distribution.
Indeed, in this case,
\RFULLRED{} performs worse than \NORED{} and \RSMARTRED{} (Figure~\ref{fig:emp_byNShards}\subref{subfig:lj_Random_t}).
\RFULLRED{}'s inferiority becomes more pronounced as $tr$ increases,
since the number of replicas that it wastefully selects increases.
When the distribution is highly skewed (Figure~\ref{fig:emp_byNShards}\subref{subfig:lj_ReDDE_t}),
few shards have high success probabilities,
hence searching additional shards becomes unproductive at some point.
Indeed, we observe a diminishing returns in \NORED{}'s recall when $tr$ increases.
In contrast, \RSMARTRED{} continues to improve by selecting replicas of high probability shards.
\COMMENT {
\begin{figure*}[tb]
	\centering     
    \subfloat[LiveJ Random]{\label{subfig:lj_Random_t}\includegraphics[scale=0.4]{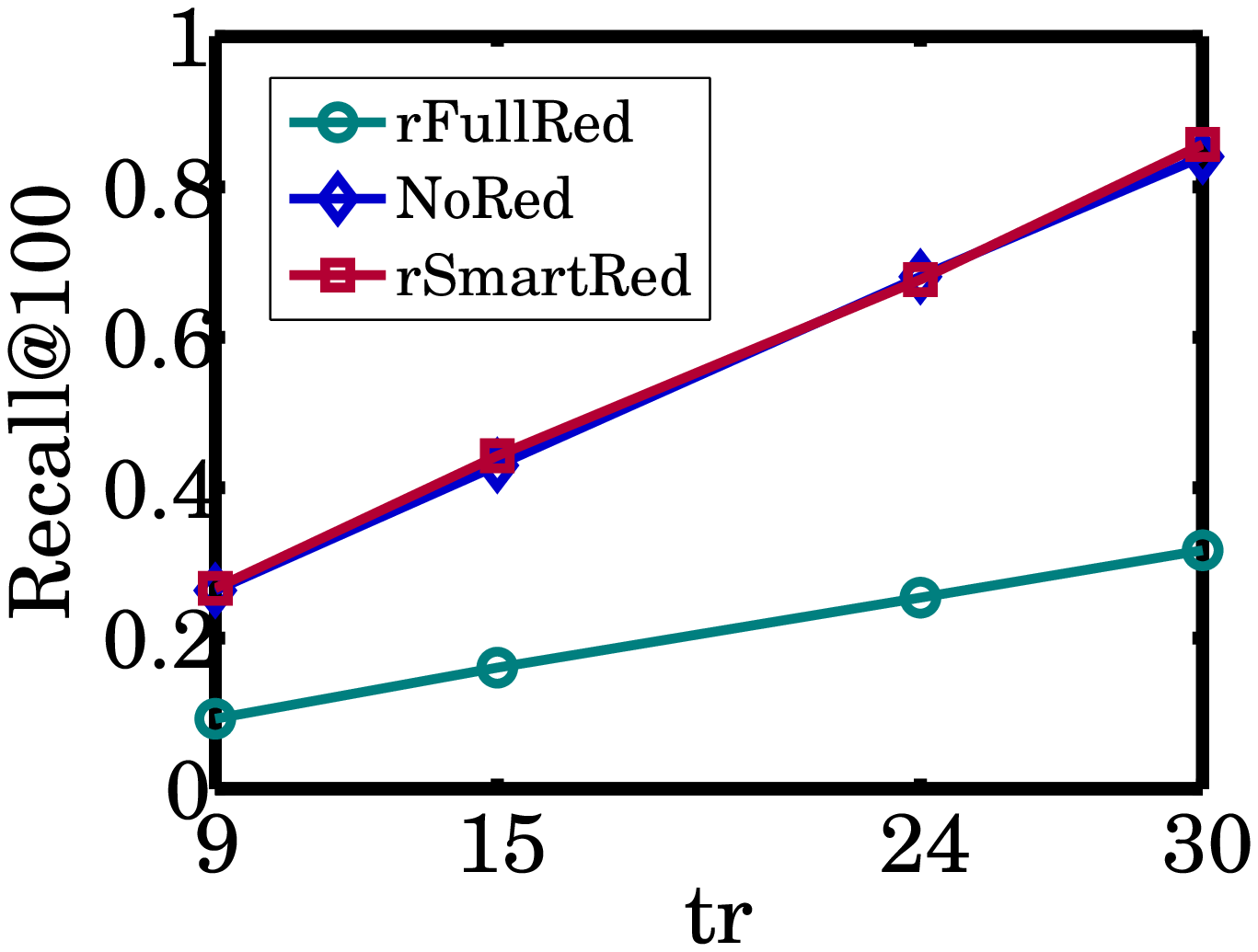}}\hspace{1em} 
    \subfloat[LiveJ CRCS]{\label{subfig:lj_ReDDE_t}\includegraphics[scale=0.4]{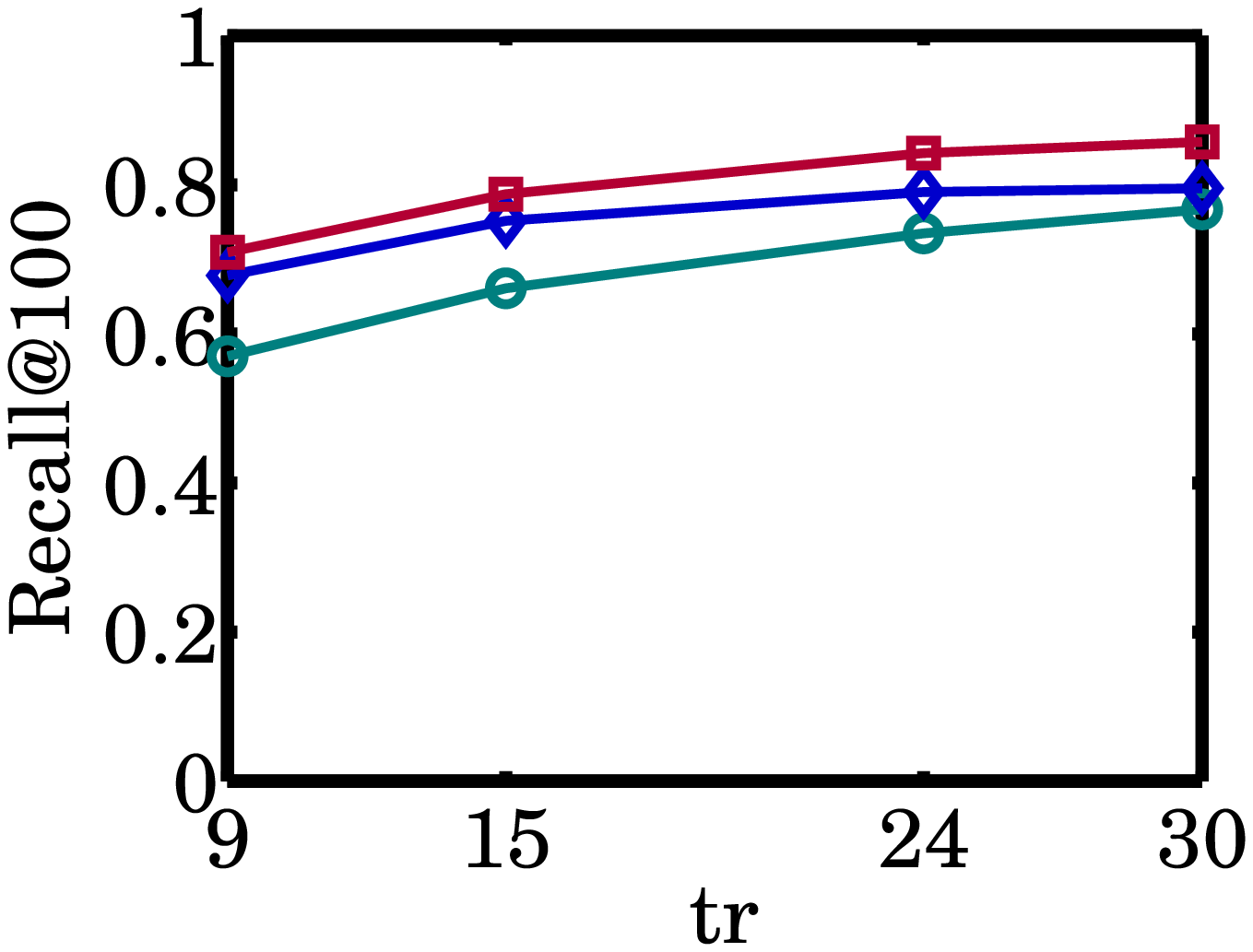}} 
    \vspace{1em}
    \caption{LiveJ Recall@100 for the three shard selection schemes for Replication as a function of the number of selected shards ($tr$).
    \RSMARTRED{} outperforms both \RFULLRED{} and \NORED{} in all scenarios.}  
	\label{fig:emp_byNShards}
    
    \subfloat[By miss probability]{\label{subfig:both_f}\includegraphics[scale=0.4]{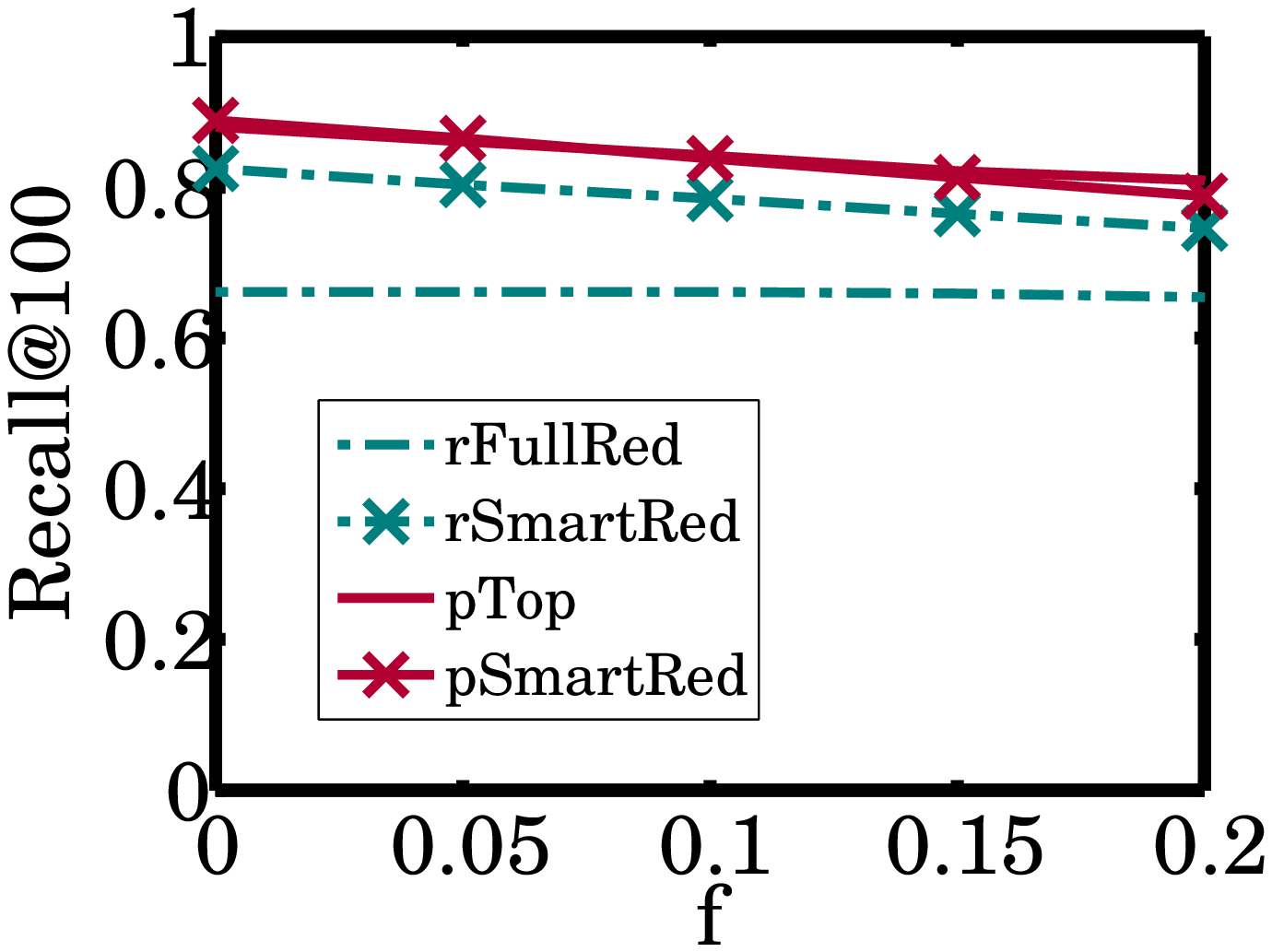}}\hspace{1em} 
    \subfloat[By number of selected shards]{\label{subfig:both_t}\includegraphics[scale=0.4]{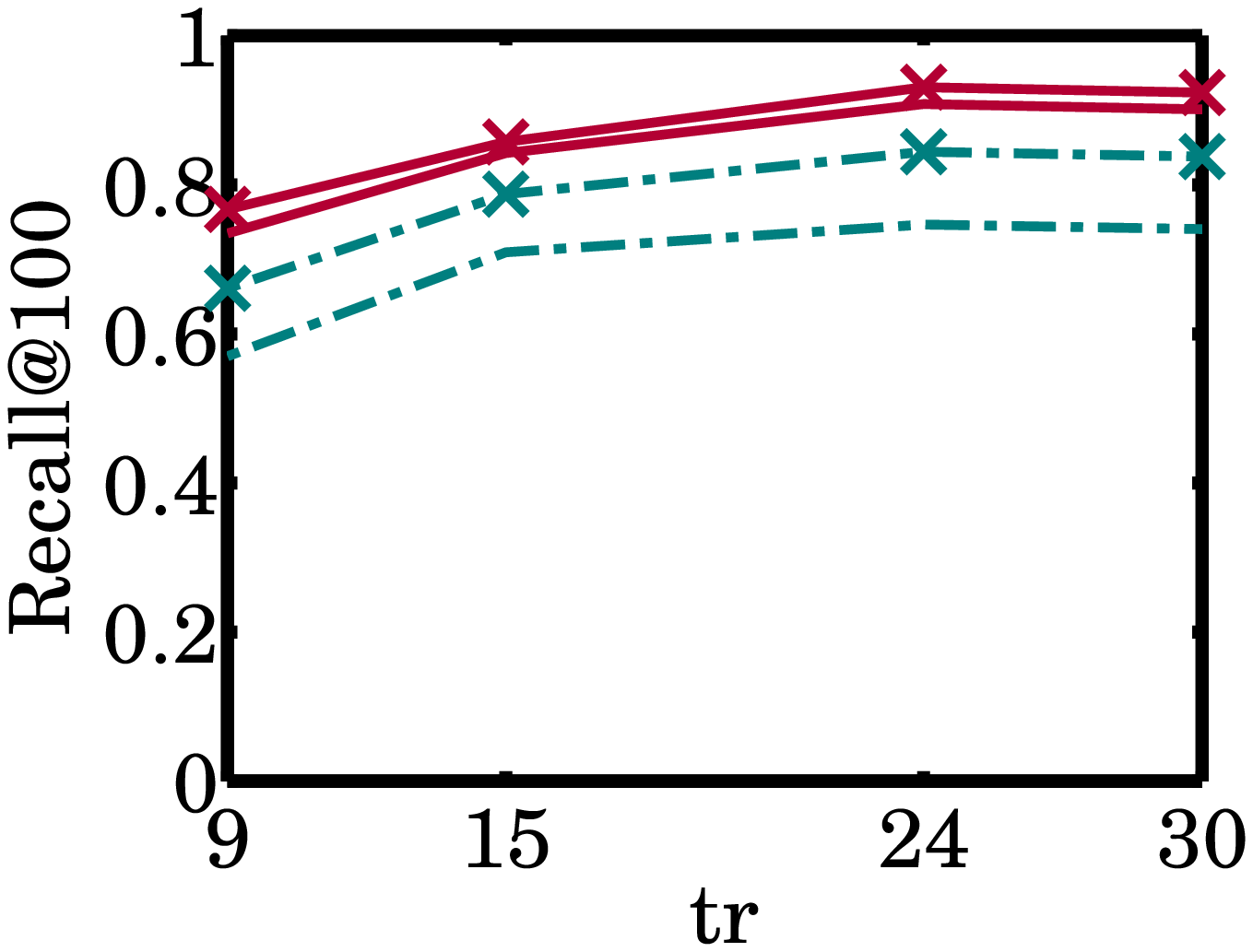}}
    \vspace{1em}
    \caption{LiveJ Recall@100 with Replication and Repartition for a skewed success probability distribution. Repartition outperforms Replication.}
    \label{fig:emp_RvsP}   
\end{figure*}
}

\subsection{Replication vs. Repartition}
\label{sec:emp_repartition}
We move on to compare between Replication and Repartition as shown in Figure~\ref{fig:emp_RvsP}.
\NORED{} is identical for both redundancy methods,
hence we omit it from the comparison.
As real deployments attempt to maintain a low miss probability,
we experiment with $0 \le f \le 0.2$.
Additionally, as real deployments attempt to
use good predictors for shard selection, 
we experiment with
the CRCS
success probability distribution.

Figure~\ref{fig:emp_RvsP}\subref{subfig:both_f} depicts recall as a function of miss probability
for a fixed $t=5$.
According to \PSMARTRED{}'s specification,
\RSMARTRED{} and \PSMARTRED{} select the same number of shards per partition.
\PSMARTRED{} achieves a statistically significant higher recall than \RSMARTRED{} thanks to using Repartition,
which confirms our analysis.
Similarly, \RFULLRED{} and \PTOP{}
select the same number of shards per partition, $t$.
Here as well, \PTOP{} achieves a statistically significant improvement over \RFULLRED{} for the same reason.
Overall, Repartition achieves a statistically significant higher recall than Replication for 
low miss probabilities and skewed success probabilities,
which reflects an important practical use case for real deployments.
In Figure~\ref{fig:emp_RvsP}\subref{subfig:both_t} we fix $f=0.1$ and vary $tr$.
We observe Repartition's superiority
for all examined $tr$ values.

\section{Conclusions}
\label{sec:conc}
We studied tail-tolerant DiS, 
which is crucial for real-world distributed search services.
We observed that tail-tolerant DiS is amenable to a non-binary availability model 
based on degradation in search quality.
We showed that in this context, Replication is not ideal for mitigating result misses,
as searching exact shard copies can be wasteful.
We introduced two strategies that better fit tail-tolerant DiS.
First, we proposed to consider miss probability
as well as each shard's probability to satisfy the query
for selecting shards.
We devised \RSMARTRED{},
an optimal shard selection scheme for Replication.
Second, we proposed Repartition, an alternative approach for applying redundancy.
Repartition constructs independent index partitions instead of exact copies,
which improves search quality over Replication in practical scenarios.

Our work considers a static index during query processing.
It would be interesting for future work to explore tail-tolerance at indexing stage
when using a dynamic index.
\balance 
\clearpage
\bibliographystyle{abbrv}
\bibliography{FTIR} 
\end{document}